\documentclass{article}
\pdfoutput=1

\usepackage{amssymb, amsthm, amsmath,mathtools, amsfonts}
\usepackage[utf8]{inputenc}

\usepackage{color}

\usepackage{geometry}                
\usepackage{graphicx}
\usepackage{amssymb}
\usepackage{epstopdf}
\def\hhref#1{\href{http://arxiv.org/abs/#1}{arXiv:#1}} 
\usepackage{hyperref}
\usepackage{amsmath}
\usepackage[dvipsnames]{xcolor}

\def\C{\mathbb{C}}

\def\phi{\varphi}
\def\DD{\mathbb{D}}
\def\CC{\mathbb{C}}

\def\D{\mathcal D}

\newtheorem{theorem}{Theorem}[section]

\newtheorem{Corollary}[theorem]{Corollary}

\newtheorem{proposition}[theorem]{Proposition}
\newtheorem{Note}[theorem]{Note}
\newtheorem{Definition}[theorem]{Definition}

\title{Noise Effects on Pad\'e Approximants and Conformal Maps}
\author{Ovidiu Costin, Gerald V. Dunne and Max Meynig \\~\\
{\small{\it Department of Mathematics, The Ohio State University, Columbus, OH 43210-1174, USA}}\\
\small{\it Department of Physics, University of Connecticut, Storrs, CT 06269-3046, USA}}

\date{}

\begin{document}
\maketitle
\begin{abstract}
We analyze the properties of Pad\'e  and conformal map approximants for functions with branch points, in the situation where the expansion coefficients are only known with finite precision or are subject to noise. We  prove that there is a universal scaling relation between the strength of the noise and the expansion order at which Pad\'e or the conformal map  breaks down. We 
 illustrate this behavior with some physically relevant model test functions and with two non-trivial physical examples where the relevant Riemann surface has complicated structure.

\end{abstract}

\vskip .5cm
\begin{quote}
  {\it Dedicated to Michael Berry on the occasion of his $80^{th}$ birthday, with sincere appreciation for a lifetime of profound and inspiring ideas and results, and also for his warmth, kindness and humour.
  }
\end{quote}

\section{Introduction}
\label{sec:intro}

In physical and mathematical applications one frequently confronts the situation in which only a finite number of terms of an expansion of the function of interest are attainable, and also these coefficients may only be known to some finite precision. Pad\'e approximants and conformal map approximants are well known  tools for the first of these problems, as they provide analytic continuation of finite-order expansions of functions beyond their radius of convergence \cite{baker, bender}. They are powerful in applications, easy to use, and have an elegant physical and mathematical interpretation in terms of electrostatics \cite{Stahl, Saff, gonchar,aptekarev,Fink,Yattselev,Costin:2021bay}. In practice, however, their accuracy is affected not only by the {\it number} of original input coefficients, but also by the {\it precision} with which these coefficients are known. Here we ask the question: how does noisy input for a finite-order series affect the accuracy of Pad\'e and conformal map approximants, for functions having a dominant branch point (or points)? This question has been studied previously for simple functions having a pole or poles \cite{froissart, bessis1, bessis2, gilewicz1, gilewicz2}, but in applications we are frequently interested in functions having branch points, and possibly quite complicated Riemann surface structure. In this case the effect of noise is much richer. We find that  the approximation accuracy is determined by a universal relation between the strength of the noise and number of input terms, when the underlying function has branch points. 
We are motivated by two primary applications: 
\begin{enumerate}
    \item Pad\'e approximation of a finite-order approximation of a {\it convergent} series, in order to probe near the dominant singularity(ies), to determine the location and nature of the branch point(s) (e.g. the {\it phase transition} point, the associated {\it critical exponent} and {\it Stokes constant}). This is a canonical problem, for example, in quantum field theory and statistical physics \cite{ZinnJustin:2002ru}. In a realistic physical application only finitely many expansion
coefficients are available and their precision is typically limited.
    
    \item Analytic continuation of a formal {\it asymptotic} expansion of a function, for which only a finite number of terms are known, and for which the coefficients are known imprecisely. This problem is regularized by a Borel transform, in which case it reduces to the first problem, in the Borel plane.\footnote{In this situation, in the absence of noise it can be shown that it is generically much more accurate to apply Pad\'e to the Borel transform than to the original asymptotic expansion (recall that Pad\'e is nonlinear so it does not commute with the transforms) \cite{Costin:2020hwg}.} A more precise analytic continuation of the Borel transform leads to a more precise analytic continuation of the physical function (the inverse Borel transform) away from the region in which the original asymptotic expansion was generated.
    
\end{enumerate}

Section \ref{sec:theorems} is devoted to the mathematical theory of noise sensitivity of Pad\'e and conformal map methods, for which we prove rigorous results for a general class of functions.
However, we also stress that there are simple practical applications of the underlying ideas in physical problems where little may be known in advance about the analytic properties of the function.
This is particularly relevant for physical applications as it is frequently the case that the behavior of the function under study is dominated by a small number of singularities (either in the physical variable or in the Borel plane). Moreover, these singularities are typically point singularities of the form 
\begin{eqnarray}
f(p)\sim (\omega-\omega_0)^\alpha A(\omega) + B(\omega) \qquad, \quad \omega\to \omega_0
\label{eq:general}
\end{eqnarray}
where $A(\omega)$ and $B(\omega)$ are analytic at $\omega_0$. If this is in the physical plane, then we are most interested in extracting information about $\omega_0$, $\alpha$ and $A(\omega_0)$, which tell us the critical point, the critical exponent and the strength of the singularity. If this is the Borel transform, then $\omega_0$ determines the strength of the leading non-perturbative effect, $\alpha$ determines the leading power-law correction, and $A(\omega)$ encodes the fluctuations about this leading non-perturbative contribution. For these cases, when the input coefficients are exactly known  there is a precise relation between the accuracy of the analytic continuation and the {\it number} of input coefficients \cite{Costin:2020hwg}. Here we extend these ideas to the situation where noise is included.

Our approach is guided by the electrostatic interpretation of Pad\'e and its natural connection with conformal maps \cite{Stahl, Saff, gonchar,aptekarev,Yattselev,Costin:2021bay}, a brief summary of which is given below in Section \ref{sec:pade-conformal}. Think of the Pad\'e poles as electrical charges in 2 dimensions, and consider the series expansion to be generated at $\omega=\infty$. Then, for a function with branch points, Pad\'e arranges the charges so that (in the limit of an infinite number of input terms) they form a skeleton-like electrical conductor with end points at the branch points, and with flexible lines of charges (wires in the limit) that deform their shape and their intersection points such that the electrical (logarithmic) capacity is minimized. Roughly speaking, it is a one-dimensional analogue of the soap bubble problem. 
This perspective leads to several deep and useful results, which we use here to quantify the accuracy of a Pad\'e approximation in the presence of noise.

The paper is organized as follows. In Section 2 we
present numerical evidence for the scaling laws governing the effects
of noise on Pad\'e and conformal map methods.  Mostly based on existing results in the literature, Section 3 makes
the connection between Pad\'e and conformal map approximants. Section 4
is devoted to the mathematical theory of noise sensitivity of Padé and
conformal map methods, for which we prove rigorous results for a
general class of functions. Section 5 discusses further physical and
mathematical applications of our analysis. Our conclusions are
summarized in Section 6.

\section{Numerical Experiments of Pad\'e ``Breakdown''}

Before coming to the general results in Section \ref{sec:theorems}, we provide some motivation based on numerical experiments which illustrate some of the key phenomena. Recall first that Pad\'e approximants do not converge pointwise, but only {\it in capacity} \cite{Stahl, Saff,gonchar,aptekarev,Yattselev} (see Sections \ref{sec:pade-conformal} and \ref{sec:theorems}). This means that care is required in defining what we mean by ``breakdown'' of the Pad\'e approximant, and how we relate it to the number of input terms and the strength of the noise. Intuitively we certainly expect that if the coefficients become too noisy then the approximation will break down. In this paper we quantify this relationship.

At one level, the breakdown of Pad\'e may be characterized by the appearance of spurious Pad\'e poles which are not related to the actual singularities of the approximated function. For the case of rational functions, in the limit that the Pad\'e order $N\to \infty$ these spurious poles appear on the unit circle with probability one \cite{gilewicz1,gilewicz2}.\footnote{This has an interesting application to filtering noisy time-series \cite{bessis1,bessis2}.}
For functions with branch points there is a richer structure of spurious poles \cite{Stahl,yamada}. Here we introduce the concept of the breakdown order, and show that it is directly proportional to the logarithm of the noise strength.

\subsection{Appearance of Spurious Pad\'e Poles Due to Noise}   
\label{sec:spurious}

We motivate our analysis with the following simple numerical experiment. The breakdown of Pad\'e can be seen as a qualitative change in the distribution of Pad\'e poles at a certain Pad\'e order, in a way that is correlated with the strength of the noise. The noise causes the appearance of arcs of poles which form a natural boundary emanating from a genuine singularity. Consider mimicking the appearance of noise by truncating the input coefficients at a certain number of digits, and plot the Pad\'e poles in the complex plane. See for example, Figures \ref{fig:one-cut-by-eye} and \ref{fig:two-cut-by-eye}, where we show the $[n, n]$ Pad\'e poles, obtained for a chosen truncation precision of the input coefficients, for the functions $(1+\omega)^{-1/9}$ and $(1+\omega^2)^{-1/9}$, respectively.  We choose these representative functions recalling that the function $(1+\omega)^{\alpha}$ characterizes the leading local behavior near an isolated branch point, and it is also the Borel transform of the function with an asymptotic expansion $F(x)\sim \frac{1}{\Gamma(-\alpha)}\sum_{m=0}^\infty (-1)^m \Gamma(m-\alpha)/x^{m+1}$, which exhibits the generic Bender-Wu-Lipatov leading large order factorial growth typical of many applications in physics \cite{leguillou}.
    \begin{figure}[h!]
      \centering{\includegraphics[scale=0.6]{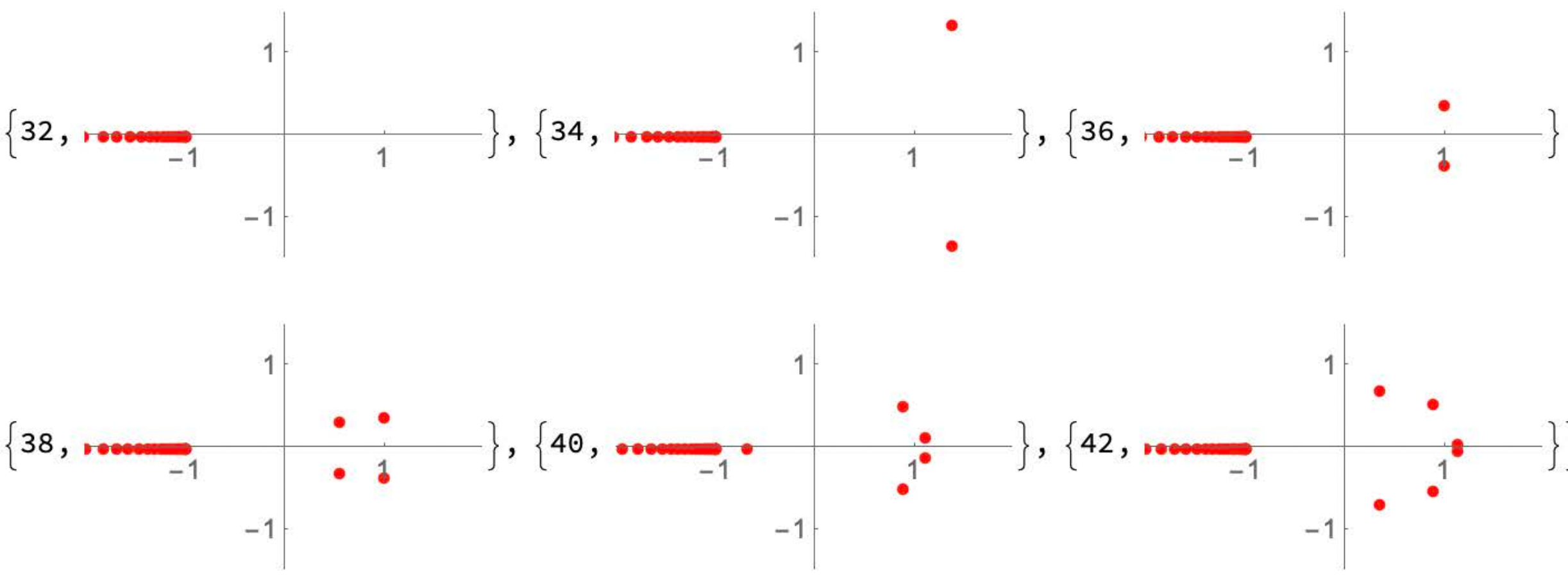}}
    \caption{Pad\'e poles for different orders, for the function $(1+\omega)^{-1/9}$, with coefficients truncated at 40 digit precision. Note that ${\rm Log}_{10}\left(4^{2\times 33}\right)\approx 40$. Spurious poles begin to appear at order $33$, and in the limit they form as arcs originating at $\omega_{\rm inf}=+1$, the point of best approximation of Pad\'e  in the absence of  noise.}
    \label{fig:one-cut-by-eye}
    \end{figure}
        \begin{figure}[h!]
            \centering{\includegraphics[scale=0.6]{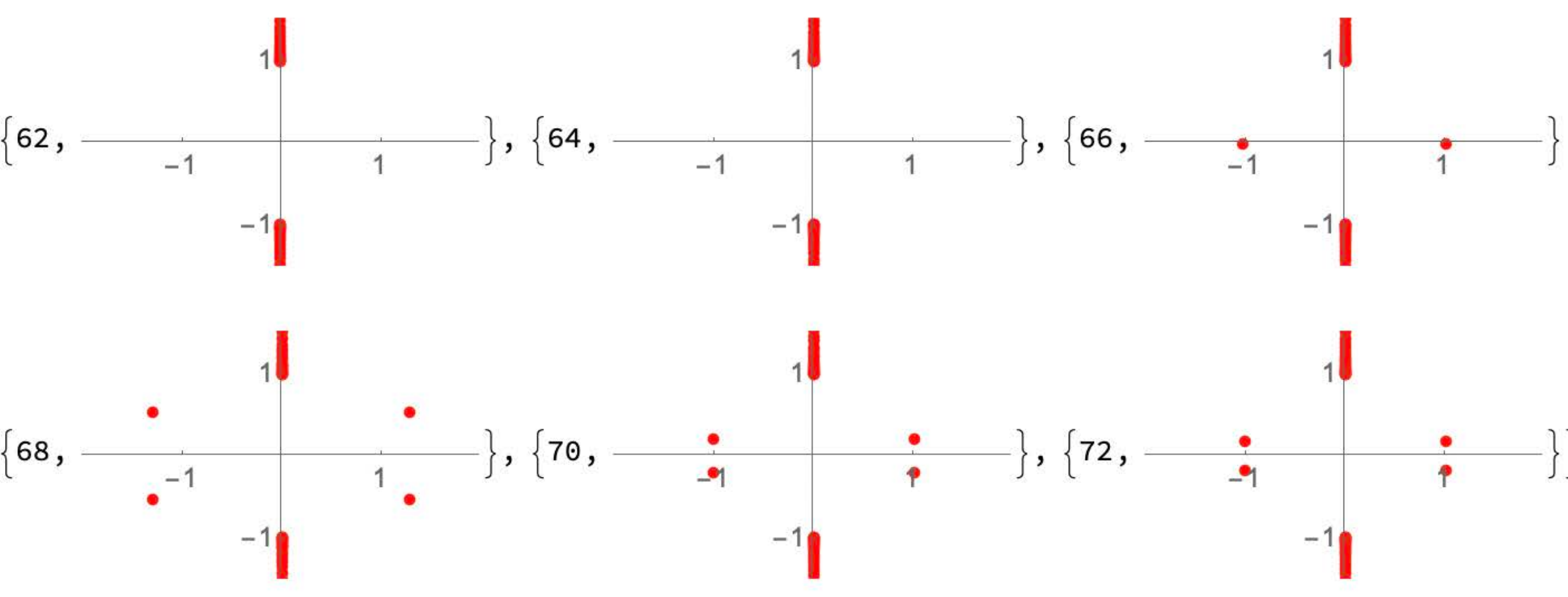}}
    \caption{Pad\'e poles for different orders, for the function $(1+\omega^2)^{-1/9}$, with coefficients truncated at 40 digit precision. Note that ${\rm Log}_{10}\left(2^{2\times 66}\right)\approx 40$. Spurious poles begin to appear at order $66$, and in the limit they form as arcs originating at $\omega_{\rm inf}=\pm 1$, the points of best approximation of Pad\'e in the absence of noise.    
     }
    \label{fig:two-cut-by-eye}
    \end{figure}
    
As is well known, with the {\it exact} rational coefficients, Pad\'e  approximates the function $(1+\omega)^{-1/9}$ with a line of poles along the negative real axis, accumulating to the branch point at $\omega=-1$, while for $(1+\omega^2)^{-1/9}$ we find two lines of poles along the imaginary axis, coming in from $\pm i \infty$ and accumulating to the branch points at $\omega=\pm i$. These lines of poles are Pad\'e's way of representing branch cuts for these functions.\footnote{ In fact, interlaced with Pad\'e zeros. Recall that since Pad\'e produces a rational function, it only has poles and zeros.} However, when the input coefficients are truncated to a chosen number of digits, one finds that as we increase the number of input coefficients, at a certain order Pad\'e begins to produce spurious poles which form an arc around an actual branch point. From numerical experiments one sees that the Pad\'e order at which these spurious poles begin to appear is correlated with the chosen digit truncation order of the input coefficients. See Figures \ref{fig:one-cut-by-eye} and \ref{fig:two-cut-by-eye}. With the chosen digit truncation order of 40, the Pad\'e order at which the spurious poles begin to appear is approximately given by $N_c=33\approx 40/(2\log_{10}(4))$ for $(1+\omega)^{-1/9}$, and by $N_c=66\approx 40/(2\log_{10}(2))$ for $(1+\omega^2)^{-1/9}$. The factor of $2$ difference between these two cases is a first hint towards a more general result, as it corresponds to the difference in the logarithmic capacities for the Riemann surfaces of these two functions. 
We make this numerical observation more precise in the following section.

The formation of spurious poles is explained in Section \ref{sec:theorems}. As the truncation errors become larger, an arc of poles starts expanding from special points $\omega_{\rm inf}$, defined below in Theorem \ref{T:4.1} (the point(s) at which Pad\'e without noise is most accurate), and in the limit these spurious poles form a circle of poles representing the natural boundary of the noise function.

\subsection{Capacity in the Presence of Noise}
\label{sec:capacity}

 In the electrostatic formulation of Pad\'e approximation (in the absence of noise), the
logarithmic capacity is an important quantity in determining the accuracy,  especially close to the point of expansion \cite{Stahl, Saff, gonchar}.  In this Section we explore the  connection between capacity and noise sensitivity.  The capacity can be estimated from the Pad\'e poles as follows. Let $\{ \omega_i \}_{i=0}^N$ be the set of poles of the diagonal $[N, N]$ Pad\'e approximant of the function $f(\omega)$. 
In the large $N$ limit, the $\omega_i$ accumulate along a set $\mathcal{C}$  of branch cuts such that $f$ is single-valued in $\CC\setminus \mathcal{C}$ (clearly, rational approximants such as Pad\'e can only converge in domains of single-valuedness of $f$). Construct the following quantity 
\begin{equation}\label{eq::TFDCapEst}
d_N(f) : = \left( \prod_{1\leq j < i\leq N } |\omega_i - \omega_j| \right)^{\frac{2}{N(N-1)}} 
\end{equation} In the $N\to\infty$ limit, $d_N(f)$ approximates the reciprocal of the logarithmic capacity of $\mathcal{C}$, \cite{Saff}:
\begin{equation}
    \lim_{N\to \infty} d_N(f) =\frac{1}{c}
    \label{eq:c}
\end{equation}
   \begin{figure}[h!]
    \centering{\includegraphics[scale=0.55]{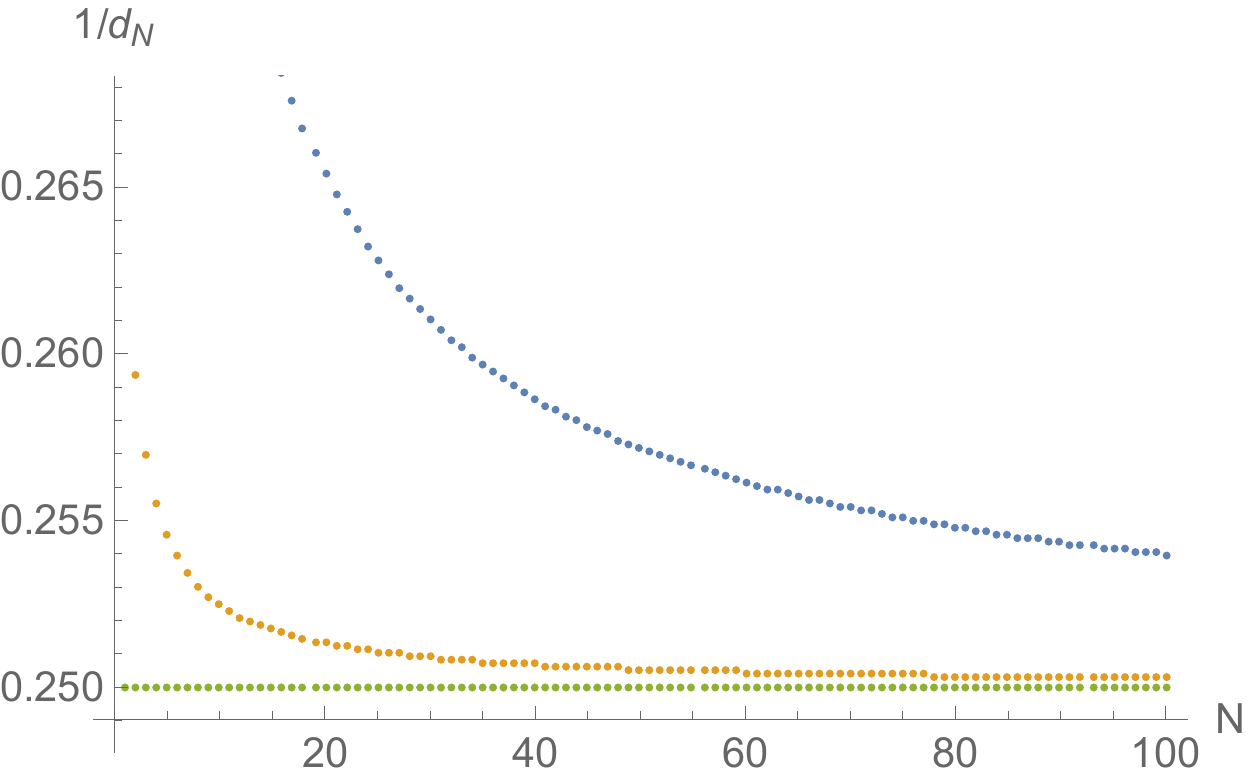}
    \includegraphics[scale=0.55]{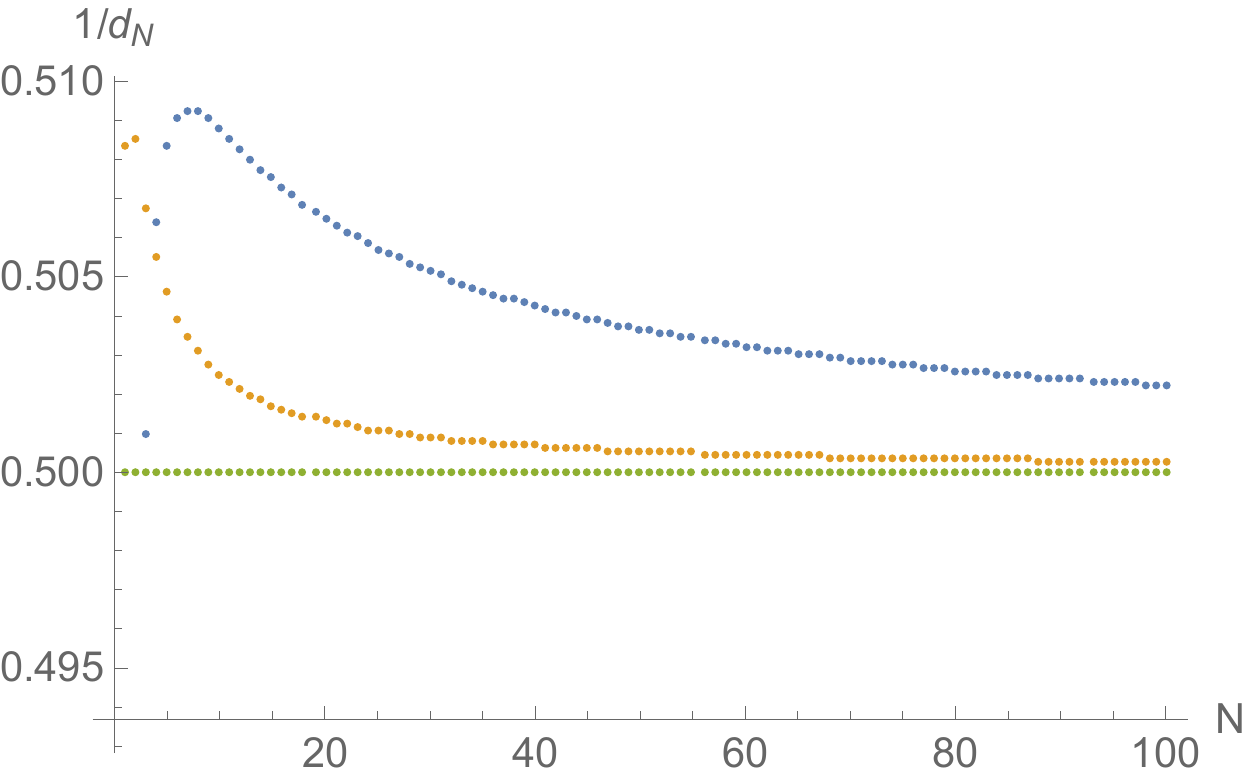}}
     \caption{Plots of the convergence of $1/d_N(f)$, defined in (\ref{eq::TFDCapEst}), to the capacity, for the functions $f(\omega)=(1+\omega)^{-1/9}$ [left] and $f(\omega)=(1+\omega^2)^{-1/9}$ [right], showing convergence to the exact values (in green) $1/4$ and $1/2$, respectively. The blue dots are $1/d_N(f)$, and the orange dots show a second order Richardson acceleration.}
        \label{fig:pade-capacity}
    \end{figure}
    The logarithmic capacity is also known as the ``transfinite diameter''  and the ``Chebyshev constant'', and can be computed in a variety of relatively simple ways  \cite{Saff,kuzmina,grassmann,ransford} and, in symmetric cases, based on continued fraction representations of Pad\'e, \cite{Damanik-Simon}. Figure \ref{fig:pade-capacity} shows the convergence of $d_N(f)$ to the known capacity values of $\frac{1}{4}$ and $\frac{1}{2}$, respectively, for the one-cut and two-cut functions $(1+\omega)^{-1/9}$ and $(1+\omega^2)^{-1/9}$ discussed in the previous section. An important result of Stahl \cite{Stahl} is that in the family of all possible cuts $\mathcal{C}$ such that $f$ is single-valued in $\CC\setminus \mathcal{C}$, the logarithmic capacity $c$ of the set $\mathcal{C}$ is minimal. It is shown in \cite{Stahl} that this minimal $c$ equals $\psi'(0)$, where $\psi(\omega)$ is the conformal map from $\mathbb{C}\setminus\mathcal{C}$ to $\DD$, normalized so that $\psi'(0)>0$. By domain monotonicity, $c<1$, see \cite{Costin:2020pcj}.

Now let us introduce noise to the coefficients, by adding to the original truncated series a series with random coefficients:
\begin{eqnarray}
f(\omega):=\sum_{k=0}^m f_k\, \omega^k \rightarrow f_\epsilon(\omega):=f(\omega)+\mathcal{N}_\epsilon(\omega); \ \ \mathcal{N}_\epsilon(\omega):=\epsilon \sum_{k=0}^m r_k\, \omega^k
\label{eq:add-noise}
\end{eqnarray}
Here $r_k$ are independent random variables distributed uniformly   in $[-1,1]$, and $0<\epsilon<1$ is a constant that characterizes the strength of the noise. Such a noise function, $\mathcal{N}_\epsilon(\omega)$, would model settings in which the Maclaurin coefficients $f_k$ are available in floating point arithmetic with a fixed number of digits. We now apply a Pad\'e approximant to the noisy truncated expansion $f_\epsilon(\omega)$ and analyze the distribution of the resulting poles. We average over many realizations of the random noise. The presence of the noise function eventually introduces new singularities in the perturbed function, which has the effect of augmenting the capacitor $\mathcal{C}$ and increasing the corresponding capacity $c$. With this introduction of noise, one finds that the convergence to the true capacity of the function, shown in Figure \ref{fig:pade-capacity}, breaks down at a certain order that is correlated with $\mathcal{C}$ in a way that is calculated in Section \ref{sec:theorems}.

This can be seen by considering the deviation from the noise-free case\footnote{Here $\epsilon' = \epsilon\, 10^{-100}$ is  a noise level small enough to have negligible impact but still to improve the computation time.}
\begin{equation}
 \Delta_N(f_{\epsilon}):=  \left\lvert d_N(f_{\epsilon})^{-1} - d_N(f_{\epsilon'})^{-1}  \right\rvert
 \label{eq:deviation}
\end{equation}
Plotting this deviation $\Delta_N(f)$  as a function of the truncation order parameter $N$, we observe a clear kink-like transition occurring at a certain order. See Figure \ref{fig:pade-noise-kink}. One can therefore estimate the breakdown point of Pad\'e as the location of this sudden kink transition.
In practice, we fix a value for the noise strength $\epsilon$, and calculate the smallest $N$ such that  $\Delta_N(f_{\epsilon})>\delta$, with $\delta$ a chosen acceptable precision threshold. In our analysis here we have chosen the error threshold $\delta=10^{-3}$. 
 \begin{figure}[h!]
    \centering{\includegraphics[scale=0.55]{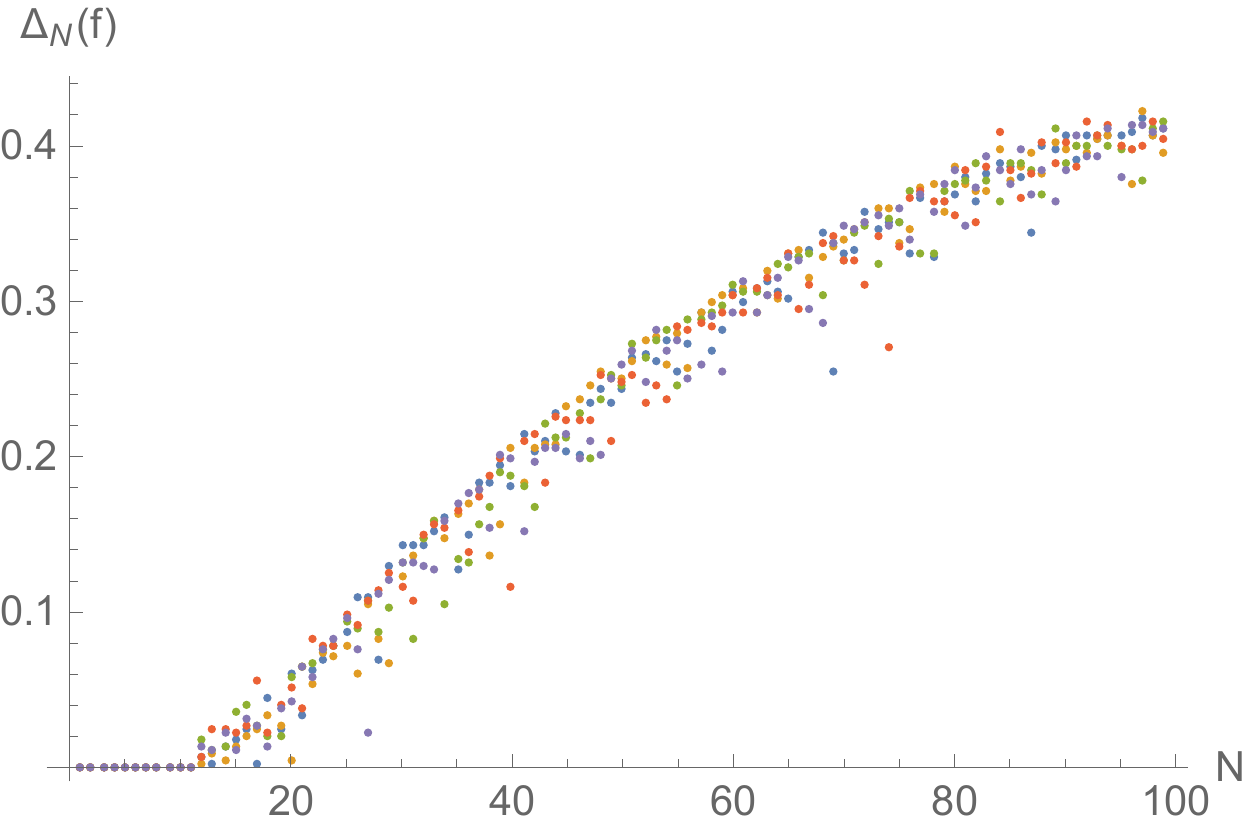}}
   \centering{\includegraphics[scale=0.55]{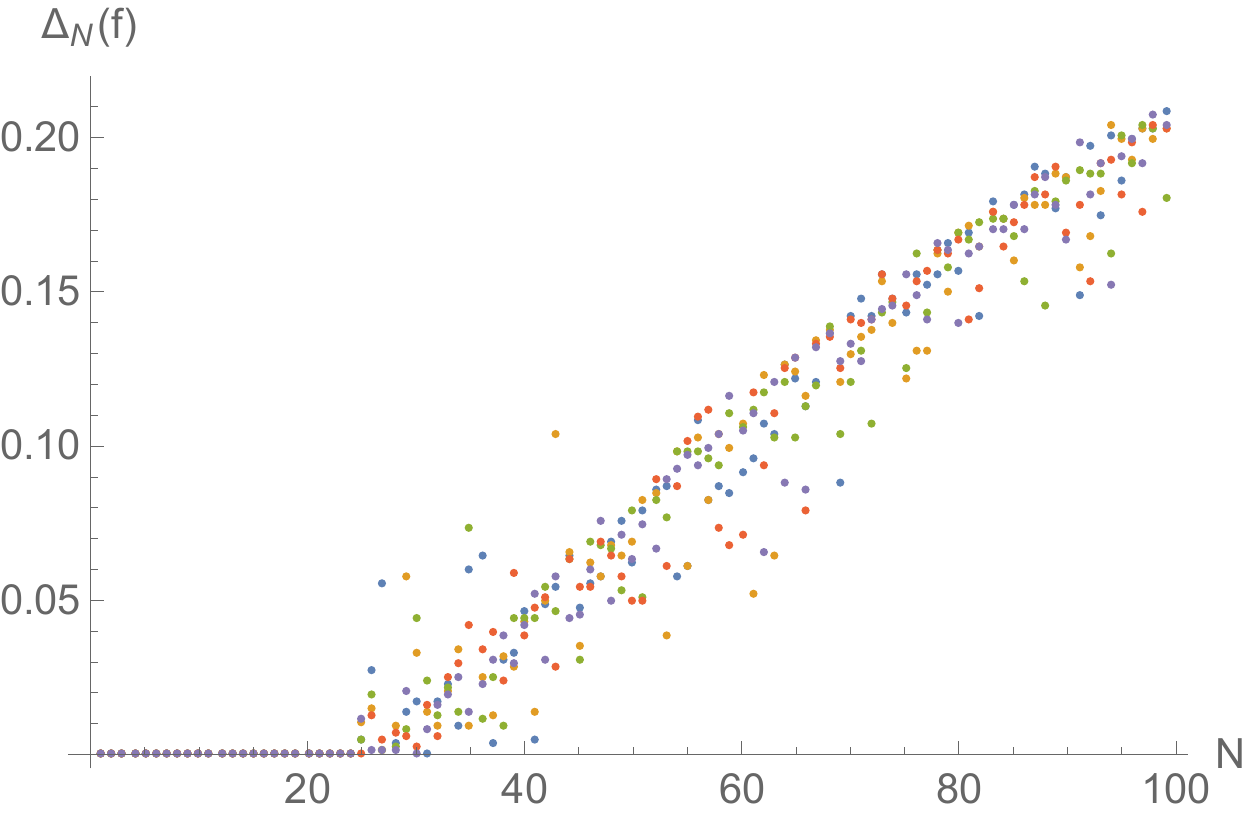}}
        \caption{Plots of the deviation in the capacity due to noise, $\Delta_N(f_{\epsilon})$, defined in (\ref{eq:deviation}), for the one-cut and two-cut functions $(1+\omega)^{-1/9}$ and $(1+\omega^2)^{-1/9}$. The different colored dots correspond to five different realizations of the random noise in (\ref{eq:add-noise}), all with a chosen noise strength $\epsilon=10^{-20}$. We see a clear breakdown once a certain truncation order $N$ is reached, and we observe that this truncation order is twice as large for the two-cut function. }
    \label{fig:pade-noise-kink}
    \end{figure}
 \begin{figure}[h!]
     \centering{\includegraphics[scale=.75]{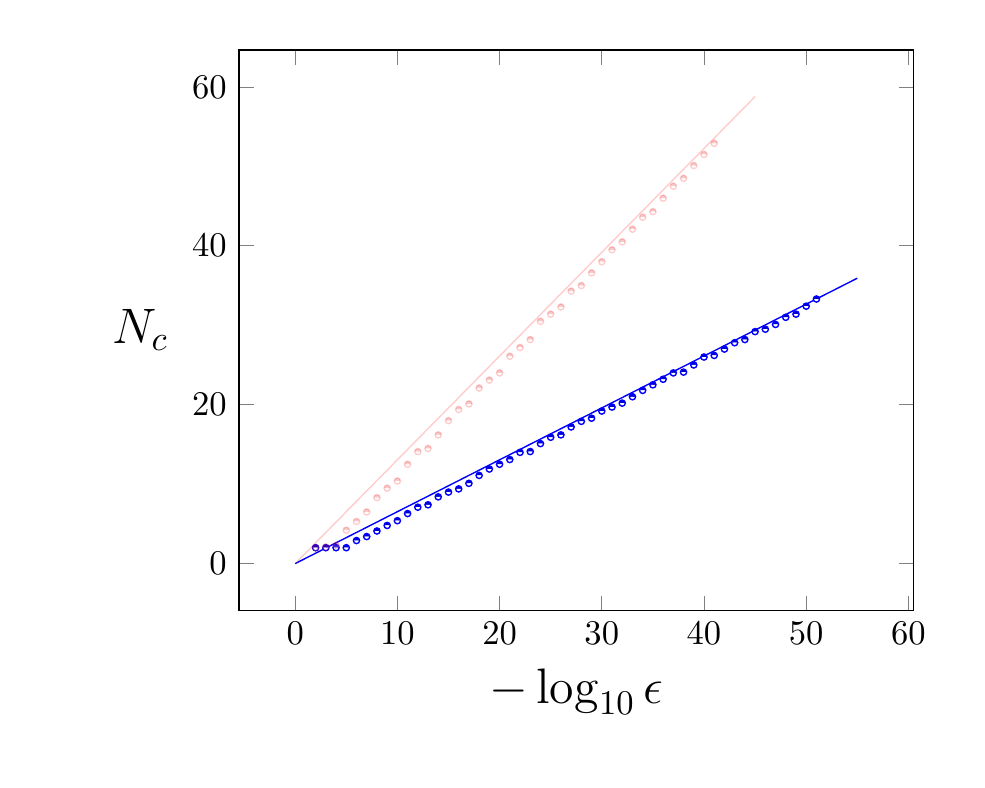}}
    \centering{\includegraphics[scale=.75]{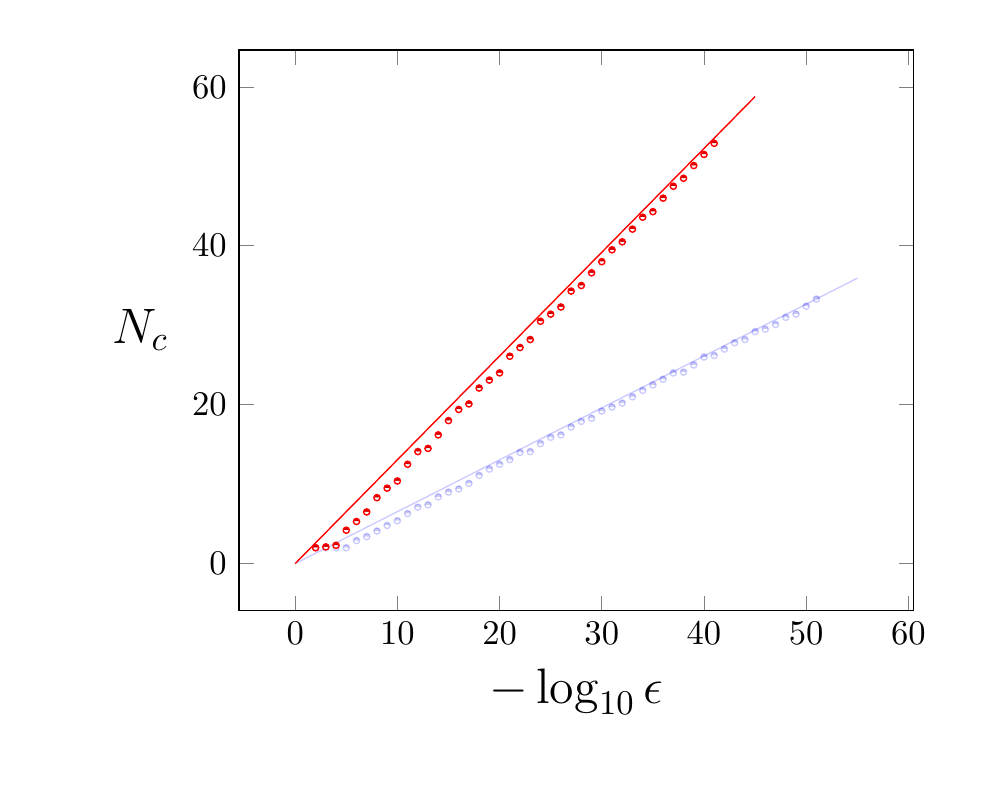}}
    \caption{Plots of the Pad\'e breakdown order $N_c$, as defined in the text, for the model functions $(1+\omega)^{-1/9}$ [left, blue dots] and $(1+\omega^2)^{-1/9}$ [right, red dots]. Note that $N_c$ is linear in the log of the noise strength $\epsilon$, as in (\ref{eq:guess}). The faint shadow lines indicate the plot for the other function, to emphasize that the slope differs by a factor of 2 in the two cases. The slope scales with the log of the capacity, as in (\ref{eq:guess}).}
    \label{fig:simple-noise-capacity}
    \end{figure}
    
We can adopt the order at which this kink transition occurs (averaged over multiple realizations of the noise $\epsilon$) as a definition of the critical order $N_c$ at which Pad\'e breaks down. With this definition, we can analyze the dependence of  $N_c$ on the noise strength $\epsilon$ in (\ref{eq:add-noise}). This is shown in Figure \ref{fig:simple-noise-capacity} for the one-cut function $(1+\omega)^{-1/9}$ and the two-cut function $(1+\omega^2)^{-1/9}$. We see a definite linear behavior in this log plot, and note that the slope in the two cases differs by a factor of 2. This matches the factor of 2 observed in the numerical experiments in Section \ref{sec:spurious} and in Figure \ref{fig:pade-noise-kink}. These numerical experiments, also repeated with other functions having different capacities, suggest a scaling relation of the form 
\begin{eqnarray}
N_c = ({\rm constant})\times \frac{\log_{10}(\epsilon)}{\log_{10}(c)}
\label{eq:guess}
\end{eqnarray}
where $\epsilon$ is the noise strength and $c$ is the capacity (recall that $c<1$). The overall constant factor appears to be universal, with an approximate value of $0.4$.
In order to explain why this is the case, and to prove  general results, we turn now to a more detailed discussion of the relation of Pad\'e to conformal maps, which provides a more analytic approach to this problem.

\section{Pad\'e and Conformal Maps}
\label{sec:pade-conformal}

The relation between Pad\'e approximants and conformal maps follows from the remarkable and intuitively useful physical interpretation of Pad\'e in terms of two-dimensional electrostatics,  known as logarithmic potential theory \cite{Stahl,Saff,gonchar,aptekarev}. Two dimensional electrostatic potentials are harmonic functions, 
for which conformal maps are of course a natural tool \cite{churchill}. We briefly review the relevant ideas and results. For excellent reviews see \cite{Saff,gonchar,aptekarev}.

The convergence of near-diagonal Pad\'e approximants to functions with branch points is a rich subject, elucidated  in the fundamental paper of Stahl \cite{Stahl}. It is interesting to note that convergence in capacity is established at this time only for functions analytic on  $\Omega\left(\CC\setminus E\right)$, for sets $E$ of zero logarithmic capacity or in domains in $\CC$ bounded by piecewise analytic arcs under a stringent symmetry condition \cite{Stahl}. The first class of functions is the relevant one here. For further developments and refinements, see \cite{Fink,Yattselev}.

\subsection{Potential Theory and Physical Interpretation of Pad\'e Approximants}
\label{sec:pade-physical}

Given a function $f(\omega)$ with branch points in $\mathbb C$, let $\D^\prime$ be any domain of single-valuedness of $f$, and let $E^\prime=\partial \D^\prime$ be its boundary. If the function $f$  has finitely many singularities, $\partial \mathcal{D}^\prime$ is a set of piecewise analytic arcs joining branch points of $f$, and some accessory points (similar to those of the Schwarz-Christoffel formula) associated with junctions  of these analytic arcs. So $E^\prime$ is the union of the chosen set of branch cuts for $f$. 

Now think of $E'$ as an electrical conductor on which we place a unit charge. The electrostatic potential on a conductor is constant, and we normalize so that the potential vanishes on $E^\prime$: $V(E^\prime)=0$. Therefore the  electrostatic capacitance of $E'$ is cap$(E')= 1/V(\infty)$. In two dimensions the potential energy per particle of a system of charges placed at the locations $E=\{\omega_i\}_{1\le i\le N}$ is:  
     \begin{equation}
       \label{eq:tot-en}
       \mathcal{E}_N(E)=-\sum_{1\le i<j\le N}\log|\omega_i - \omega_j|=: - \frac{N(N-1)}{2} \log \delta_N(E)
     \end{equation}
The minimal potential energy is attained with the charges at the equilibrium positions, known as 
Fekete points \cite{Saff}. In the limit $N\to \infty$, $\delta_N(E)$  gives the harmonic capacity, the exponential of the usual capacity with respect to infinity. Compare with \eqref{eq::TFDCapEst}. The key relations to Pad\'e are as follows:  
   \begin{enumerate}

\item The electrostatic minimization process results in the {\em minimal capacitor}, and asymptotically it coincides with the poles of Pad\'e (where the Pad\'e is constructed from an expansion about infinity rather than about zero, in order to match the electrostatic interpretation). In other words, in the $N\to\infty$ limit the Pad\'e poles are placed along an electrical conductor $E$, and the complement $\mathcal{D}$ of $E$ is the domain of convergence of Pad\'e \cite{Stahl}. 

  \item For $\omega\in \mathcal{D}$, the Green's function is related to the potential as $|G_{\mathcal{D}}(\omega)|=e^{-V(\omega)}$. In fact $|G_{\mathcal{D}}(\omega)|=\psi_{\mathcal{D},\infty}(\omega)$, where $\psi_{\mathcal{D},\infty}$ is the conformal map from $\mathcal{D}$ to the unit disk $\DD$, seen from infinity with $\psi_{\mathcal{D},\infty}(\infty)=0$ (\cite{Stahl}). This conformal map can be recovered, in the limit $N\to \infty$, from the harmonic function $|G_{\mathcal{D}}|$, see Proposition \ref{P:3.1} below.
 Summarizing for diagonal Pad\'e from Theorem 1 by Stahl \cite{Stahl}, 
\begin{enumerate}
\item For any $\epsilon>0$ and any compact set $V\subset \mathcal{D}\setminus\{\infty\}$ we have
 \begin{equation}
  \label{eq:limcap1}
  \lim_{N\to\infty}\text{cap}\{\omega\in V|(f-[N, N]_f)(\omega)>(G_{\mathcal{D}}(\omega)+\epsilon)^{2N}\}=0
\end{equation}
\item If $f$ has branch points, which occurs iff $G_{\mathcal{D}}\ne 0$, then for any compact set $V\subset \mathcal{D}\setminus\{\infty\}$ and any $0<\epsilon\le \inf_{\omega\in V}G_{\mathcal{D}}(\omega)$ we have
 \begin{equation}
  \label{eq:limcap2}
  \lim_{N\to\infty}\text{cap}\{\omega\in V|(f-[N, N]_f)(\omega)<(G_{\mathcal{D}}(\omega)-\epsilon)^{2N}\}=0
\end{equation}

\end{enumerate}

\item In this sense, Pad\'e effectively ``creates its own conformal map'' and its own domain $\mathcal{D}$. Geometrically,  join all the branch points of $f$ by a perfectly conducting, connected, infinitely flexible wire $W$ in such a way that the function $f$ is single-valued in the complement of $W$. The wire generically has further junction nodes besides the branch points. Deform the wire until the capacitance of the final, extremal, wire $E$ with respect to infinity is minimized.

  \item  The equilibrium measure $\mu$  on $E$ is the equilibrium density of charges on $E$. As $N\to\infty$, the poles of the near diagonal Pad\'e approximants place themselves (except for a set of zero capacity) close to $E$, and Dirac masses placed at these poles converge in measure to $\mu$  \cite{Stahl}. 
  
  \item The numerators and denominators of Pad\'e approximants are orthogonal polynomials, in a generalized sense, along arcs in the complex domain, but therefore without a bona-fide Hilbert space structure. According to \cite{Stahl}, this is the ultimate source of capacity-only convergence, and of the appearance of spurious poles. Spurious poles can be eliminated, cf. \cite{Stahl}, p. 145, (8), after which convergence is uniform.

  \end{enumerate}
  
       \begin{proposition}\label{P:3.1}{\rm 
        Assume that spurious poles have been eliminated and convergence is uniform \cite{Stahl}.
Let $f$ be analytic in the unit disk and have branch points in $\CC$. Then, for large $N$ and $\omega\in\mathcal{D}$,
 \begin{equation}
         \label{eq:partialD}
         |f(\omega)-[N, N]_f(\omega)|^{1/2N} =|\psi(\omega)|(1+o(1))
       \end{equation}
      where $\psi$ is the conformal map from the domain of analyticity of the Pad\'e  approximants to $\DD$. In fact, with an appropriate choice of branch, 
     \begin{equation}
         \label{eq:partialD2}
        e^{-i\lambda}\left( f(\omega)-[N, N]_f(\omega)\right)^{1/2N} =\psi(\omega)(1+o(1))
      \end{equation}  for some phase $ e^{-i\lambda} $.}
      
      \end{proposition}
      
     \begin{proof}
         Choose a  disk $D\subsetneq \DD$ around the origin. Since $\psi$ is conformal and $\psi(0)=0$, we have $\psi(z)\ne 0$ on $\partial D$. Let $\kappa=\max_{\partial D} |\psi |(\min_{\partial D} |\psi|)^{-1}$. Note also that, inside $\mathcal{D}$ we have $|\psi|<1$. Choose a $\delta>0$ small enough so that on $\partial D$ we have $|\psi|(1+\delta)<1$. Let $\varepsilon=\kappa\delta$. 
      Fixing a radius $R$, combining \eqref{eq:limcap1} and \eqref{eq:limcap2} we get for that for large $N$ \eqref{eq:partialD} holds on $\partial D$. 
      By the maximum principle, the inequality holds in $D$. For the second part, we note that $|\psi|$ is harmonic, and $|f-[N, N]_f|=|\psi|^{2N}(1+o(1))$, while $\left( f(\omega)-[N, N]_f(\omega)\right)$ is analytic.
     \end{proof}

\subsection{Pad\'e and Conformal Maps in the Presence of Noise: Two Simple Examples} 
\label{sec:pade-conformal-noise}

An important consequence of these results connecting Pad\'e with logarithmic potential theory and conformal map methods is that we can now understand how and why noise affects a Pad\'e approximant: in the large $N$ limit the noisy input coefficients effectively propagate through the Pad\'e algorithm by composition of the conformal map with the original (noisy) series. This composition of series introduces a geometric growth factor and massive cancellations which amplify the effect of the noise. Therefore the problem can be re-cast as the analysis of the effect of a conformal map on a noisy series, and this can be quantified precisely, providing us with a sharp quantitative estimate of the relation between the noise and the number of terms before Pad\'e breaks down.

It is instructive to show how this works for the generic case of a function with one dominant branch point, which we can normalize to be at $\omega=-1$. Then in the large $N$ limit the effect of noise on the Pad\'e approximant is given by composition of the truncated noisy series $f_\epsilon(\omega)$ in (\ref{eq:add-noise}) with the one-cut conformal map:
\begin{eqnarray}
\omega=\frac{4z}{(1-z)^2}=:\varphi(z) \quad \longleftrightarrow\quad z=\frac{\sqrt{1+\omega}-1}{\sqrt{1+\omega}+1}=:\psi(\omega)
\label{eq:1cut-map}
\end{eqnarray}
which maps the cut $\omega$ plane to the interior of the unit disk in the conformal $z$ plane. This maps the branch point $\omega=-1$ to $z=-1$, the origin $\omega=0$ to $z=0$, and the upper (lower) edge of the cut $\omega\in (-\infty, -1]$ is mapped to the upper (lower) half of the unit circle $|z|=1$.
Therefore
\begin{eqnarray}
\left(\frac{4z}{(1-z)^2}\right)^m =4^m \sum_{k=0}^\infty z^{m+k} \begin{pmatrix} k+2m-1 \cr k\end{pmatrix}
\label{eq:expansion}
\end{eqnarray}
and so the composition yields:
\begin{eqnarray}
\left(f_\epsilon \circ \varphi\right)(z)&=& \sum_{m=0}^\infty 4^m \sum_{k=0}^\infty \left(a_m +\epsilon\, r_m\right) \left(\frac{4z}{(1-z)^2}\right)^m
\nonumber\\
&=& \sum_{m=0}^\infty z^m \sum_{k=0}^m \left(a_{k} +\epsilon\, r_{k}\right) 4^k \begin{pmatrix} m+k-1 \cr m-k\end{pmatrix}
\label{eq:composition}
\end{eqnarray}
The variance of $[z^m]\left(f_\epsilon \circ \varphi\right)(z)$, the coefficient of $z^m$,  is 
\begin{eqnarray}
\sigma^2(m)&=& \langle \left([z^m]\left(f_\epsilon \circ \varphi\right)(z)- \langle [z^m]\left(f_\epsilon \circ \varphi\right)(z)\rangle\right)^2\rangle
\label{eq:variance}
\end{eqnarray}
Because the random noise averages to zero, $\langle r_m\rangle=0$, the variance reduces to
\begin{eqnarray}
\sigma^2(m)&=& \epsilon^2 \langle \left([z^m]\left(f_\epsilon \circ \varphi\right)(z)- \langle [z^m]\left(f_\epsilon \circ \varphi\right)(z)\rangle\right)^2\rangle
 \nonumber\\
&=& \epsilon^2 \bigg\langle \left( \sum_{k=0}^m r_{k} 4^k \begin{pmatrix} m+k-1 \cr m-k\end{pmatrix}\right)^2\bigg\rangle
\nonumber\\
&=& \frac{\epsilon^2}{3}  \sum_{k=0}^m 4^{2k} \begin{pmatrix} m+k-1 \cr m-k\end{pmatrix}^2
\label{eq:variance2}
\end{eqnarray}
At large $m$, the summand is strongly peaked around $k\approx \frac{m}{\sqrt{2}}$, and so the sum can be evaluated by a straightforward Euler-Maclaurin analysis. We find the large $m$ estimate
\begin{eqnarray}
\sigma^2(m)\approx \frac{\epsilon^2}{3} \frac{\left(\sqrt{2}+1\right)^{4m}}{2^{1/4} \sqrt{2\pi m}}\left(1+o\left(\frac{1}{m}\right)\right)
  \label{eq:sigma1m-approx}
\end{eqnarray}
which is in excellent agreement with the variance in (\ref{eq:variance2}) even at modest values of $m$: see Figure \ref{fig:sigma1m}.
\begin{figure}[htb]
\centering\includegraphics[scale=.6]{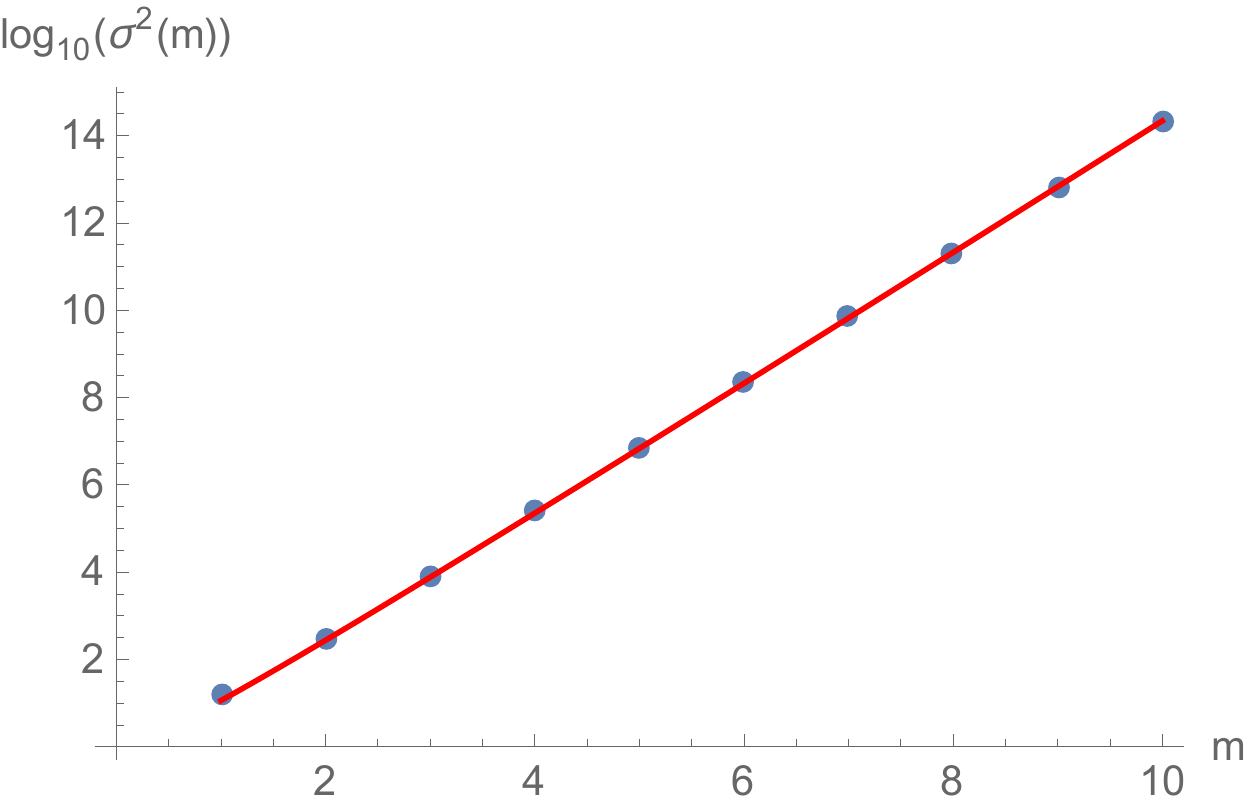}
\caption{A log  plot of $\sigma^2(m)$ in (\ref{eq:variance2}) [blue dots] compared with the Euler-Maclaurin leading large $m$ approximation in (\ref{eq:sigma1m-approx}). For this comparison plot we have suppressed the unimportant common overall $\frac{\epsilon^2}{3}$ factor in (\ref{eq:variance2}) and (\ref{eq:sigma1m-approx}). }
\label{fig:sigma1m}
\end{figure}

To connect back to the Pad\'e approximant, recall that for a diagonal $[N, N]$ Pad\'e approximant, the original series is expanded to order $m=2N$, and the optimal truncation of the re-expansion after the conformal map is to truncate the $z$ expansion also at this order \cite{Costin:2020pcj}. Therefore, in the variance we identify $m=2N$. We arrive at a sharp estimate for the slope in the empirical identification (\ref{eq:guess}). 
\begin{eqnarray}
N_c \approx \frac{-\log_{10}(\epsilon)}{4\, \log_{10} \left(1+\sqrt{2}\right)}
\label{eq:nc1}
\end{eqnarray}
The corresponding conformal map with $M$ symmetrically placed radial cuts emanating to infinity from branch points at $\omega^M=-1$ is \cite{Costin:2020pcj,Kober} (with the natural branch choices):
\begin{eqnarray}
\omega=\frac{2^{2/M}\,z}{(1-z^M)^{2/M}} =: \varphi_M(z) \quad \longleftrightarrow\quad z=\sqrt[M]{\frac{\sqrt{1+\omega^M}-1}{\sqrt{1+\omega^M}+1}}=:\psi_M(\omega)
\label{eq:Mcut-map}
\end{eqnarray}
The analysis is therefore identical, with $\omega$ replaced by $\omega^M$ and $z$ replaced by $z^M$, so the only real difference is the scaling of the highest term in the expansion. We therefore find
\begin{eqnarray}
N_c \approx \frac{-M \log_{10}(\epsilon)}{4\, \log_{10} \left(1+\sqrt{2}\right)}
\label{eq:ncM}
\end{eqnarray}
This agrees with the scaling factor of $2$ difference between the one-cut and two-cut cases found in Sections \ref{sec:spurious} and \ref{sec:capacity}. 

Recalling that for this symmetric configuration of $M$ radial cuts, the capacity is $c_M= \psi_M^\prime(0)=4^{-1/M}$ \cite{Costin:2020pcj,Costin:2021bay}, we can re-express (\ref{eq:ncM}) in a more suggestive form:
\begin{eqnarray}
N_c(M)=\frac{\log_{10}(\sqrt{2})}{\log_{10}(1+\sqrt{2})}\, \frac{\log_{10}(\epsilon)}{\log_{10}(c_M)}
\approx (0.3932...)\, \frac{\log_{10}(\epsilon)}{\log_{10}(c_M)}
\label{eq:result-M}
\end{eqnarray}
This agrees very well with the numerical fit of $\approx 0.4$ of the slope in (\ref{eq:guess}), for the data in the numerical experiments shown in Figure \ref{fig:simple-noise-capacity}, and it incorporates the correct scaling with the capacity.

\section{Mathematical  Theory of the Effect of Noise on Pad\'e and Conformal Maps}
\label{sec:theorems}

In this Section we show that when noise is introduced there is, generically, universality in the breakdown of conformal map methods and of associated Pad\'e approximations.The two are related due to the key result of Stahl \cite{Stahl}, described above in (\ref{eq:limcap1})-(\ref{eq:limcap2}). This has the remarkable consequence that at large order, the leading error of Pad\'e is expressed solely in terms of the conformal map $\psi(\omega)$ and its inverse $\varphi(z)$, not in terms of the function $f(\omega)$. As in the case of conformal map methods, this dramatically simplifies the problem, as we can decouple the analysis from the (unknown) function $f$ and concentrate on a conformal map with the same singularity structure.
We introduce a general mathematical characterization of accuracy breakdown in terms of properties of the relevant conformal map. Recall that the practical application of Pad\'e and associated conformal maps is to extrapolate a truncated Taylor series beyond its radius of convergence. The introduction of noise reduces the accuracy of  Pad\'e and conformal map methods, by various natural measures of accuracy listed in the definition \ref{D:4.2} below.

{\bf Notation.} Let $\mathcal{D}$ be a domain, more generally a Riemann surface $\Omega$, in which the function  of interest, $f:\Omega\to\C$, is analytic. We denote by $\psi=\psi(\omega)$ the conformal map, or more generally the uniformization map, from $\Omega$  to $\DD$, and $\varphi=\varphi(z)$  its inverse. We normalize such that the leading singularity of $f$ has $\omega_0=1$, and so $f$ is analytic in $\DD$ and on $\Omega$, and we write
\begin{equation}
  \label{eq:ser-f}
  f(\omega)=\sum_{j=0}^\infty f_j\, \omega^j
\end{equation}
We denote by $[N, N]_f$ the {\em diagonal Pad\'e approximant} $P_N/Q_N$ of $f$.\footnote{Our results extend to near-diagonal Pad\'e approximants which, together with the diagonal ones, are the most relevant in applications.} Here $P_N, Q_N$ are defined as usual to be the unique pair of polynomials of degree $N$ normalized so that the coefficient of $\omega^N$ in $P_N$ is $1$, and such that the  Maclaurin coefficients of order up to and including $2N$ of $[N, N]_f$ agree with those of $f$ \cite{baker,bender}.
For a function $g$ we denote by $g_{[m]}$ its  Maclaurin polynomial of of order $m$.

We write the expansion of $\phi(z)$ as
\begin{equation}
  \label{eq:ser-phi}
 \varphi(z)=\sum_{k=1}^\infty b_k z^k
\end{equation}
Noting that $\varphi(\DD)=\Omega$, and that $f$ is analytic in $\Omega$, it follows that the composition $f\circ\phi$ is analytic in $\DD$. Hence the series of $f\circ\varphi$,
\begin{equation}
  \label{eq:compo1}
  (f\circ \varphi)(z)=\sum_{j=0}^\infty f_j\varphi^j(z) : =\sum_{k=0}^\infty c_k\,  z^k
\end{equation}
converges in the unit disk $\DD$.
The effect of noise on Pad\'e reduces to the effect of noise on composition with a conformal map as discussed in the following.

  \begin{Definition}\label{D:4.2}  {\rm 
        In reconstructing $f$ from its truncated expansion $f_{[m]}$ by using a conformal map or relatedly, by Pad\'e, both using $f\approx(f\circ\phi)_{[k]}\circ \psi$, $k\ge n$, the noise-induced breakdown of approximation  can be defined in a number of ways:
   \begin{enumerate}
    \item The first value of $k$ for which the coefficient $c_k$ in \eqref{eq:compo1} becomes inaccurate (by some given measure).

    \item The first value of $k$ for which there is a $z\in\DD$  for which the approximation $(f\circ\phi) (z)\approx (f\circ\phi)_{[k]}(z)$  becomes inaccurate.

    \item For a fixed $z\in\DD$,  the first value of $k$ for which the approximation $(f\circ\phi)(z)\approx (f\circ\phi)_{[k]}(z)$  becomes inaccurate.

    \end{enumerate}
    As we will see, (1.) and (2.) above are roughly equivalent, while (3.) provides more local information.
  
 }   \end{Definition}
 
 We quantify the loss of accuracy in Corollary \ref{C:4.5} and, much more sharply,  in Theorem \ref{T:4.1} below.  These are generalizations of the results derived in Section \ref{sec:pade-conformal-noise} for the simple representative conformal maps for functions with one branch cut, or a symmetric set of radial branch cuts, and show that in general the important mathematical object is the conformal map, either the explicit one being used, or the one that Pad\'e effectively constructs (recall the discussion of the electrostatic interpretation of Pad\'e in Section \ref{sec:pade-physical}).

\begin{Note}{\rm 

  \begin{enumerate} 
  \item  Since $[N, N]_f$ has the same Maclaurin coefficients as $f$ up to and including order $2N$, for any $\varphi$, $f\circ \varphi$ and $[N, N]_f\circ \varphi$ have the same Maclaurin coefficients up to and including the $2N^{th}$ order.

  \item Assume that spurious poles have been eliminated from the Pad\'e approximant \cite{Stahl}. In \cite{Costin:2020pcj} it is shown that the accuracy of approximation of $[N, N]_f$ is lower than, and of the same order of magnitude as, that of  $(f\circ \varphi)_{[2N]}$.

\item    The approximation provided by Pad\'e is that of $(f\circ \varphi)_{[2N]}+\sum_{k\ge 2N+1} a_k\, z^{k}$, where the $a_k$ are generally different from the Maclaurin coefficients of $(f\circ \varphi)$ and make Pad\'e less accurate than $(f\circ \varphi)_{[2N]}$. By any sensible measure of accuracy, Pad\'e breaks down earlier than the associated conformal map, albeit not significantly earlier.
  \end{enumerate}
 } \end{Note}

  \begin{Note}\label{N:caveat}
       {\rm The question of the behavior of a series with ``random coefficients'' near the circle of convergence is a delicate one which goes back at least to Borel who raised this problem and hinted that series with arbitrary coefficients must have natural boundaries \cite{Borel1}. This statement was made precise  by Paley and Zygmund (see, e.g., \cite{kahane}), and one of the most general results is the Ryll-Nardzewski theorem \cite{Ryll}:
       
     \begin{theorem}\label{T:P-Z}{\rm 
       Let symmetric $X_i$ be random variables such that  $\limsup_{N\to \infty} |X_N|^{1/N}=1$, $X_N=X_N(\omega)$ and $F(z)=\sum_{k=0}^\infty X_k z^k$. Then $\partial\mathbb{D}$ is almost surely a natural boundary of $F$.
     } \end{theorem}
 }\end{Note}
The following result about the behavior of the coefficients $n_k$ of $(\mathcal{N}\circ\phi)(z)$ follows straightforwardly from Theorem \ref{T:P-Z}.
\begin{Corollary}\label{C:4.5}
  Let $r_i$ be independent random variables as in Section \ref{sec:capacity}, let $\Omega\supset \mathbb{D}\ne \CC$ be a simply connected domain in $\CC$ or more generally a Riemann surface containing $\DD$ on its first Riemann sheet, and assume that $S^1\cap \partial \Omega$ consists of finitely many points $\mathcal{B}=\{\omega_1,...,\omega_M\}$ (singular points of $f$). Let $\psi:\Omega\to \DD$ be the conformal or uniformization map of $\Omega$ to $\DD$ and $\phi=\psi^{-1}$. Let $\mathcal{N}_\epsilon$ 
  be the noise function introduced in \eqref{eq:add-noise} and
  \begin{equation}
    \label{eq:eqnn}
    n(z):=(\mathcal{N}_\epsilon\circ\phi)(z)=\sum_{k=0}^\infty n_k z^k
  \end{equation}
  Then,
  \begin{enumerate}
    
  \item The unit $\omega$ circle $S^1$ in $\Omega$ is a natural boundary of $\mathcal{N}_\epsilon$.
    
  \item The curve $\mathfrak{b}:=\{z:|\phi(z)|=1\}=\psi(S^1)$ is  a natural boundary of $\mathcal{N}_\epsilon\circ\phi$, and it is piecewise analytic with nonanalytic points at $\psi(\omega_i)$.
    
  \item Let  $z_{\rm inf} \in \mathfrak{b}$ be such that $|z_{\rm inf}|=\min_{z\in\mathfrak{b}}|z|$. Then $|z_{\rm inf}|<1$.

  \item For each realization of the noise variables $r_i$ we have, with probability one, 

    $$\limsup_k|n_k|^{\frac1k}=|z_{\rm inf}|^{-1} $$
 \end{enumerate}

\end{Corollary}
\begin{proof}

  1. This is simply Theorem \ref{T:P-Z}.

  2. Since $\phi$ is a biholomorphism, if $z_0\in \DD$ is a point of analyticity of $\mathcal{N}_\epsilon\circ\phi$, then $\omega_0\in S^1$ is a point of analyticity of $\mathcal{N}_\epsilon$, implying that $\mathfrak b$ is a natural boundary for $\mathcal{N}_\epsilon\circ\phi$.

  3. If $\omega\in S^1\setminus \mathcal{B}$, then $\omega\in \Omega\cap S^1$, hence, by domain monotonicity (see \cite{Costin:2020pcj})),  $\psi(\omega)\in \DD$. On each analytic arc of the curve $\mathfrak{b}$, there are points of minimum $z_j$  of $|z|$, where, by the above, $|z_j|<1$ (since inside the analytic arcs $z\in \DD$, and $|z|=1$ at the endpoints of the arcs which belong to $\partial\Omega$). We simply let $z_{\rm inf}$ be a point of absolute minimum of $|z|$ among these $z_j$. 
  
  4. By 1., $z_{\rm inf}$ is a singular point of $\mathcal{N}_\epsilon\circ\phi$, the closest singular point to the origin. The result follows from the $k$th root test.
\end{proof}
We have the following,  sharper for our purposes, result that does not rely on Theorem \ref{T:P-Z}. We write  $z_{\rm inf}=\psi(\omega_{\rm inf})$ for one or more $\omega_{\rm inf}$ and $\omega_{\rm inf}=e^{i\theta_{\rm inf}}$.
    \begin{theorem}\label{T:4.1}
      For large $k$, $n_k$ is a random variable of zero average and standard deviation
    \begin{equation}
      \label{eq:sharpbound}
   \sigma(n_k)=A\, \epsilon\, k^{-\frac14}|z_{\rm inf}|^{-k}(1+o(1))
 \end{equation}
 where 
 $$A= 3^{-\frac12}(2\pi)^{\frac34}|\psi'(\omega_{\rm inf})|\, [\Re\alpha(\omega_{\rm inf})]^{-\frac14}$$
Here  $\alpha=\frac{\psi''}{\psi}-\left(\frac{\psi'}{\psi}\right)^2$ is assumed to be nonzero at $\theta_{\rm inf}$, which is generic. (Note that, by conformality, $\psi'\ne 0$.).  Corollary \ref{C:4.5} holds. 
    \end{theorem}

\begin{proof} We have
  \begin{multline}
    \label{eq:eqnk}
    n_k=\frac{1}{2\pi i}\oint_{\mathfrak{b}_-}\frac{(\mathcal{N}_\epsilon\circ\phi)(s)}{s^{k+1}}ds=\frac{1}{2\pi i}\oint_{C_-}\frac{\mathcal{N}_\epsilon(\omega)}{\psi(\omega)^{k+1}}\psi'(\omega)d\omega
\\= \sum_{j=0}^\infty   \epsilon r_j \frac{1}{2\pi i}\oint_{C_-}\frac{\omega^j}{\psi(\omega)^{k+1}}\psi'(\omega)d\omega=\sum_{j=0}^\infty   \epsilon r_j \frac{1}{2\pi i}\oint_{S^1_\circ}\frac{\omega^j}{\psi(\omega)^{k+1}}\psi'(\omega)d\omega=\sum_{j=0}^\infty   \epsilon r_j C_{j,k}
\end{multline}
where $S^1_\circ$ is a curve avoiding the singular points of $\psi$ on $S^1$ through small arccircles in $\DD$. Since at a singular point we have $|\psi|=1$, while $1>|z_{\rm inf}|=\inf_{S^1}|\psi|$,  the contribution of the arcs is relatively exponentially small. We also note that $1/\psi$ is bounded on $S^1$ since $\psi(S^1)=\mathfrak{b}\subset \overline{\DD}\setminus\{0\}$.

By conformality,  $\psi(\omega)\ne 0$ on $S^1$. We  write $\psi=e^{\ln|\psi|+i\arg\psi}=e^{F+i\Phi}$, and change variable as $\omega=e^{i\theta}$. The integrand then takes the form
\begin{equation}
  \label{eq:exponent}
  \exp\Big\{-(k+1)F(\theta)+i[(j+1)\theta-(k+1)\Phi(\theta)]\Big\}\psi'(e^{i\theta})
\end{equation}
We assume that $\omega_{\rm inf}$ is unique. The general case  follows by superposition  of the contributions of the finitely many $\omega_{\rm inf}$.

We apply the saddle point method in an abstract way since $\psi$ is general. By the minimum condition,  $\frac{d}{d\theta}F |_{\theta_{\rm inf}}=0$. For every $k$ there is a range of $j$ such that the integrand of $C_{j,k}$ has a saddle point near  $\theta_{\rm inf}$. The  rest of the $C_{j,k}$ are relatively exponentially small, since  the exponent in \eqref{eq:exponent} is analytic, and if $\theta_{\rm inf}$ is not a saddle point in a  $C_{j,k}$, then there is an analytic steepest descent line through  $\theta_{\rm inf}$ along which the contour can be further pushed.
The condition that a saddle is placed near $\theta_{\rm inf}$ is
$$\Phi'(\theta_{\rm inf})=\frac{j_0+1}{k+1}\approx \frac{j+1}{k+1} $$
There always exists such a $j_0>0$, since $\frac{d}{d\theta}\Phi >0$ (this inequality follows from by conformality: as $\omega$ traverses $S^1$ in a positive direction $\psi(\omega)$ traverses $\mathfrak{b}$ in a positive direction). We write $j=j_0+(j-j_0)$ and apply the saddle point method at $\theta=\theta_{\rm inf}$. Since $\int_{-\infty}^\infty x^{2m}e^{-kx^2+ibx\sqrt{k}}dx=P_m(b)e^{-b^2}k^{-m-\frac12}(1+o(1))$, where $P_m$ is a polynomial independent of $k$, a calculation shows that the asymptotic expansion in inverse powers of $k$ near a saddle is valid as long as $k^{-1}(j-j_0)^2\ll \ln k$. The $C_{j,k}$ with $j$ beyond this range are much smaller than any within the range, as explained above.

Let $\alpha=F''+i\Phi''=\frac{\psi''}{\psi}-\left(\frac{\psi'}{\psi}\right)^2$; since $F$ has a minimum at $\theta_{\rm inf}$, we have $\alpha\ge 0$; we assume that the minimum is generic and $\alpha>0$. For $k^{-1}(j-j_0)^2\ll \ln k$  we have
\begin{equation}
  \label{eq:cj}
  C_{j,k}(1+o(1))=
\sqrt{\frac{2\pi}{\alpha(\theta_{\rm inf})k}}e^{-\frac{(j-j_0)^2}{2k \alpha(\theta_{\rm inf}) }}z_{\rm inf}^{-k}\psi'(\omega_{\rm inf}) 
\end{equation}
and therefore
\begin{equation}
  \label{eq:eqnk2}
  n_k(1+o(1))=\epsilon\, z_{\rm inf}^{-k}\sqrt{\frac{2\pi}{\alpha(\theta_{\rm inf})k}}\, \psi'(\omega_{\rm inf}) \sum_{j\ge 0}r_je^{-\frac{(j-j_0)^2}{2k \alpha(\theta_{\rm inf}) }}
\end{equation}
Using leading order Euler-Maclaurin to evaluate the sum we get
\begin{multline}\label{eq:varn}
  \sigma^2 (n_k)(1+o(1))=\frac{\epsilon^2}{3} \frac{2\pi|z_{\rm inf}|^{-2k}}{|\alpha(\theta_{\rm inf})|k}|\psi'(\omega_{\rm inf}) |^2\sum_{j\ge 0}e^{-\frac{(j-j_0)^2}{2k  }\frac{\Re\alpha(\theta_{\rm inf})}{|\alpha(\theta_{\rm inf})|^2}}\\=\frac{\epsilon^2}{3} \frac{2\pi|z_{\rm inf}|^{-2k}}{|\alpha(\theta_{\rm inf})|k}\sqrt{2\pi}|\psi'(\omega_{\rm inf}) |^2\sqrt{\frac{|\alpha(\theta_{\rm inf})|^2}{\Re\alpha(\theta_{\rm inf})}}\sqrt{k}=\frac{\epsilon^2}{3} (2\pi)^{\frac32}|\psi'(\omega_{\rm inf}) |^2[\Re\alpha(\theta_{\rm inf})]^{-\frac12}k^{-\frac12}|z_{\rm inf}|^{-2k}
\end{multline}
The last part of the theorem follows from \eqref{eq:varn}, since the probability
$$P\Big(\limsup_k \sum_{j\ge 0}e^{-\frac{(j-j_0)^2}{2k  }\frac{\Re\alpha(\theta_{\rm inf})}{|\alpha(\theta_{\rm inf})|^2}}=0\Big)=0$$
as it is easy to check. \end{proof}

\begin{Corollary}
\label{C:4.7}
  For a given error threshold $\delta$,
\begin{enumerate}
  \item  The breakdown condition on the coefficient $n_k$ is
  \begin{equation}
    \label{eq:eqbr}
    A\, \epsilon\,  k^{-1/4}|z_{\rm inf}|^{-k}(1+o(1))\gtrapprox\delta
  \end{equation}
  
\item The condition of  breakdown of approximation at $\omega=\phi(z)$, $|\omega|>1$, is
  \begin{equation}
    \label{eq:eqbz}
      A\, \epsilon\,  k^{-1/4} \left|\frac{z}{z_{\rm inf}}\right|^{k}\left|\frac{1}{1-z_{\rm inf}/z}\right|\gtrapprox \delta
  \end{equation}
  (where $|\omega|>1$ implies $|z|> |z_{\rm inf}|$).

  In the limiting case $|z|=1$, 2. reduces, up to a constant, to the condition in 1. 
  
  \end{enumerate}
 
\end{Corollary}
           \begin{Note}{\rm 
           \begin{enumerate}
 
 \item
            Beyond the breakdown order, Pad\'e places more and more poles in arcs on $S^1$. The points $z_{\rm inf}$ are the nearest singularities to the origin of the noise function $\mathcal{N}_\epsilon\circ\phi$ in the conformal disk,  ``as seen'' by Pad\'e approximants. Hence, when the noise becomes strong enough, Pad\'e adds these points (or, rather, points close to them and slightly farther from the origin) to the list of singularities, and, as the effects of noise intensify, new poles spread out from these to form, in the limit, a circle of ``noise poles'' corresponding to the natural boundary of $\mathcal{N}_\epsilon$.
            
            \item
            This result explains the pattern of spurious poles in Figures \ref{fig:one-cut-by-eye} and \ref{fig:two-cut-by-eye}. The spurious poles form arcs near the points $\omega_{\rm inf}$, which in Figure \ref{fig:one-cut-by-eye} is the point $\omega_{\rm inf}=+1$, and in  Figure \ref{fig:two-cut-by-eye} are the two points $\omega_{\rm inf}=\pm1$.
            
             \end{enumerate}

            }
         \end{Note}
\begin{Note} {\rm 
 \begin{enumerate}
 
 \item Pad\'e provides an efficient method to estimate the key quantity $z_{\rm inf}$ in Corollary \ref{C:4.5}, 3., even with only a limited number $N$ of input coefficients: one simply computes the difference $[N, N]_f(\omega)-[N-1, N-1]_f(\omega)$, for $\omega\in \partial\mathbb D$ and looks for the smallest value.
 
 \item     For Pad\'e, one could also estimate $z_{\rm inf}$  empirically, with high accuracy, as follows. One constructs a known function $g$ with branch points at the tips of $\mathcal{D}$ and analytic in $\mathcal{D}$, calculates Pad\'e approximants of sufficient order for $g$, and measures the error of approximation along the unit disk in $\omega$. Then $z_{\rm inf}$ follows from \eqref{eq:partialD} above. 

\item
  It follows from the same analysis that the region where conformal map or Pad\'e approximants are guaranteed to be insensitive to noise is the unit disk in $\Omega$, the same as the domain of convergence of the original series. Any extrapolation beyond this domain of convergence is eventually affected by noise.
  
  \item
For diagonal $[N, N]$ Pad\'e we identify $m=2N$ and we therefore have a general mathematical characterization of breakdown, expressed in a form analogous to (\ref{eq:guess}) and (\ref{eq:result-M}):
  \begin{eqnarray}
 N_c^{\rm inf}=\frac{\log_{10}(\epsilon)}{2\log_{10}(z^{\rm inf})}
  \label{eq:final}
  \end{eqnarray}
  Thus, the proportionality factor is most naturally identified with $z_{\rm inf}$, which is a property of the conformal map associated with Pad\'e.
  
  \item
  Note that for the configurations of $M$ symmetric branch points, discussed in Section \ref{sec:pade-conformal-noise}, which are often relevant in physical applications, this general breakdown condition agrees with (\ref{eq:result-M}), because for the map (\ref{eq:Mcut-map}) we have
  \begin{eqnarray}
z_M^{\rm inf} := \inf_{\theta\in [0, 2\pi)}\left[ \psi_M(e^{i \theta}) \right]=\psi_M(1)= \frac{1}{(\sqrt{2}+1)^{2/M}}
\label{eq:psi-inf-M}
\end{eqnarray}
Therefore
  \begin{eqnarray}
N_c^{\rm inf} (M)&=& \frac{-M\log_{10} (\epsilon)}{4\log_{10}(\sqrt{2}+1)}
\label{eq:comparison-inf}
\end{eqnarray}
  in agreement with (\ref{eq:ncM}).
  \end{enumerate}
  }
  \end{Note}

\section{Physical Applications}

In nontrivial applications we typically do not know the full Riemann surface structure of the function being approximated. However, 
in many physical and mathematical applications the function's behavior is dominated by  finitely many singularities, often just one or two. In such situations, even approximate information about these dominant singularities can be used to construct accurate approximations to the function that are significantly more precise than the original series expansion. 

But now we ask what happens in the presence of noise. Our main result is that the key quantity in relating the number of terms at which Pad\'e breaks down to the strength of the noise is the conformal map produced by Pad\'e in the large $N$ limit.
Importantly for applications, this map only depends on the {\it locations} of the singularities, so the relation between $N_c$ and the noise strength can be estimated using even approximate information about the singularity locations.

To illustrate the generality of this result, we now study the numerical analysis of two non-trivial examples coming from physical and mathematical applications, where we do not know the exact conformal map, but we can construct an approximate map based on the {\it leading} singularities. These are applications in which it is possible to generate terms of an asymptotic expansion, with {\it exact} rational coefficients, and the divergent formal series can be used to explore the singularity structure of the corresponding Borel plane using combinations of Pad\'e approximants and conformal and uniformizing maps. But in both cases, the Borel plane has an intricate multi-sheeted Riemann surface structure, so the underlying functions are much more complicated than the simple one-cut and two-cut functions used in the numerical experiments in Sections \ref{sec:spurious} and \ref{sec:capacity}. Nevertheless, we show that estimates based on their leading singularities match very closely the actual behavior of Pad\'e in the presence of noise.

\subsection{Renormalization in Quantum Field Theory}
\label{sec:phi36}

Perturbation theory in quantum mechanics and quantum field theory (QFT) is generically divergent, with factorially growing coefficients 
\cite{leguillou}. In QFT it is generally difficult to generate many terms of a perturbative expansion,  
and frequently such an expansion has coefficients that are only approximate. In this Section we choose a particular computation of the anomalous dimension $\gamma(a)$ in an asymptotically free conformal theory, scalar $\phi^3$ theory in 6 dimensional spacetime. This is a well-studied theory \cite{Lipatov,Fisher,Mckane,Bonfim,Borinsky}, and one for which high orders of perturbation theory are accessible using the Kreimer-Connes Hopf algebraic approach to renormalization \cite{Connes:1999yr}. Broadhurst and Kreimer showed that in this approach the perturbative expansion of the anomalous dimension $\gamma(a)$ is characterized by a quartically nonlinear third order ODE for $\gamma(a)$  \cite{Broadhurst:1999ys}. The resurgent trans-series structure of this function has recently been analyzed in detail in \cite{Borinsky:2021hnd,Borinsky:2022knn}, revealing an intricate Borel Riemann surface structure. On the first sheet there is a single {\it dominant} Borel branch point singularity with exponent $\frac{1}{12}$, in addition to two further resonant collinear Borel singularities, and all three of these singularities are repeated in integer multiples. Here we show that if noise is introduced to this computation, the result is dominated by the leading Borel singularity, so the slope relating the breakdown order $N_c$ to the logarithm of the noise strength, as in (\ref{eq:final}), can be well approximated by the one-cut situation.
 \begin{figure}[h!]
 \centering{\includegraphics[scale=1]{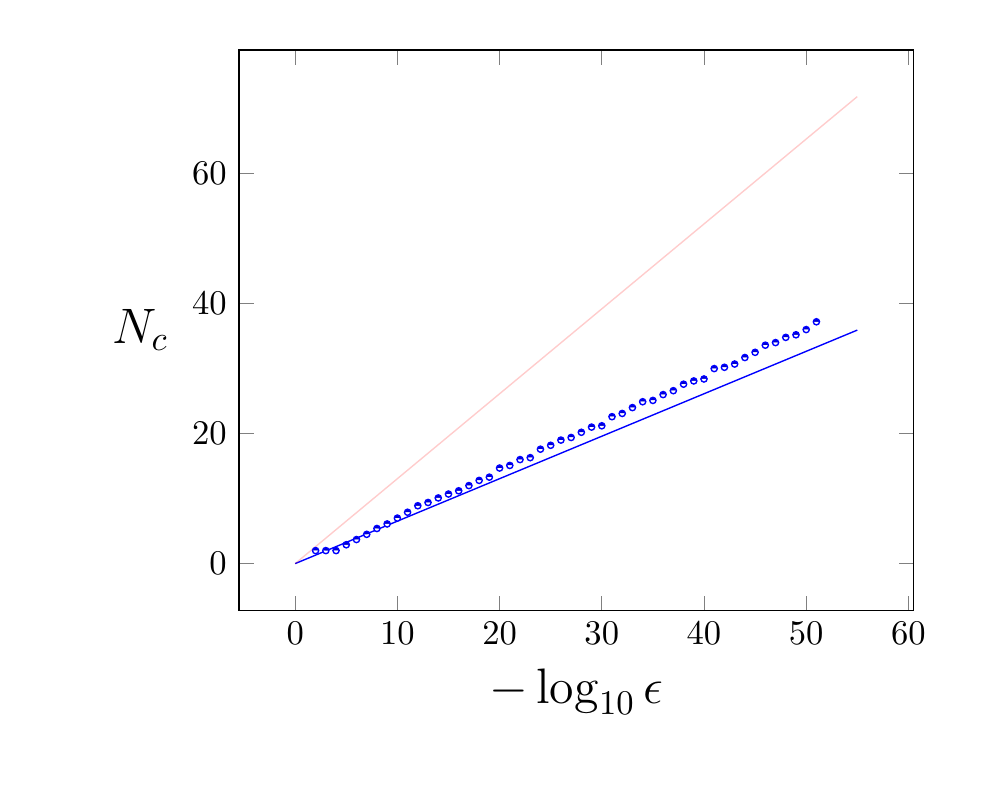}}
    \caption{This plot shows the critical truncation order $N_c$ at which Pad\'e breaks down in the presence of noise, as a function of the logarithm of the noise strength, for the Borel transform (\ref{eq:anom}) of the anomalous dimension in the Hopf algebraic analysis of the $\phi^3$ scalar quantum field theory in 6 dimensions \cite{Broadhurst:1999ys,Borinsky:2021hnd}. The blue dots show the average of multiple realizations of the random noise, and the blue line shows the general one-cut scaling relation in (\ref{eq:nc1}). Contrast this with the faint red line which shows the slope for the two-cut case.}
    \label{fig:phi36-noise}
    \end{figure}

The function under consideration is the Borel transform (see equation (29) in \cite{Borinsky:2021hnd})
\begin{eqnarray}
G(\omega)=6\sum_{n=0}^\infty (-1)^n \frac{A_n}{12^n} \frac{\omega^n}{n!}
\label{eq:anom}
\end{eqnarray}
where the coefficients $A_n$ appear as \href{https://oeis.org/A051862}{A051862} in the Online Encyclopedia of Integer Sequences. Truncating this expansion at a given order, we can apply Pad\'e to analytically continue the Borel transform, in order to obtain a resummation of the divergent perturbative expansion of the anomalous dimension. If we then introduce noise to the expansion coefficients, as in (\ref{eq:add-noise}), then we observe that the Pad\'e approximation to $G_\epsilon (\omega)$ breaks down at a truncation order $N_c$ that depends on the strength of the noise. Figure \ref{fig:phi36-noise} plots this critical truncation order $N_c$ as a function of the logarithm of the noise. It is quite remarkable that such a drastic approximation of only considering the effect of the location of a single {\it dominant} branch point singularity captures the general trend quite accurately.
 
\subsection{Tritronqu\'ee Solution to Painlev\'e I}
\label{sec:painleve1}

The Painlev\'e equations generate solutions known as the ``nonlinear special functions'', with a wide range of applications in physics and in mathematics \cite{clarkson}. Asymptotic expansions of these functions can be described in terms of resurgent transseries \cite{costin-book}, and their Borel transforms have a rich Riemann surface structure, with infinitely many sheets \cite{Costin:2020pcj}. As a concrete example we consider the Borel transform of the perturbative expansion of the tritronqu\'ee solution to Painlev\'e I, which arises in physical applications in matrix models of 2d gravity \cite{DiFrancesco:1993cyw}. This special solution  $F(x)$ undergoes nonlinear Stokes transitions in the physical domain when ${\rm Arg}(x)$ is an integer multiple of $\frac{2\pi}{5}$. In the Borel plane the tritronqu\'ee Borel transform function $f(\omega)$ has two infinite towers of collinear  Borel singularities, at all integer multiples of a $\pm$ pair. The analysis of \cite{Costin:2019xql} shows that this solution can be accurately analytically continued into the complex $x$ plane, starting from an asymptotic expansion generated for $x\to+\infty$, even crossing into the Dubrovin pole region $\frac{4\pi}{5} < {\rm Arg}(x) < \frac{6\pi}{5}$ \cite{dubrovin,costin-dubrovin}.
\begin{figure}[h!]
    \centering{\includegraphics[scale=1]{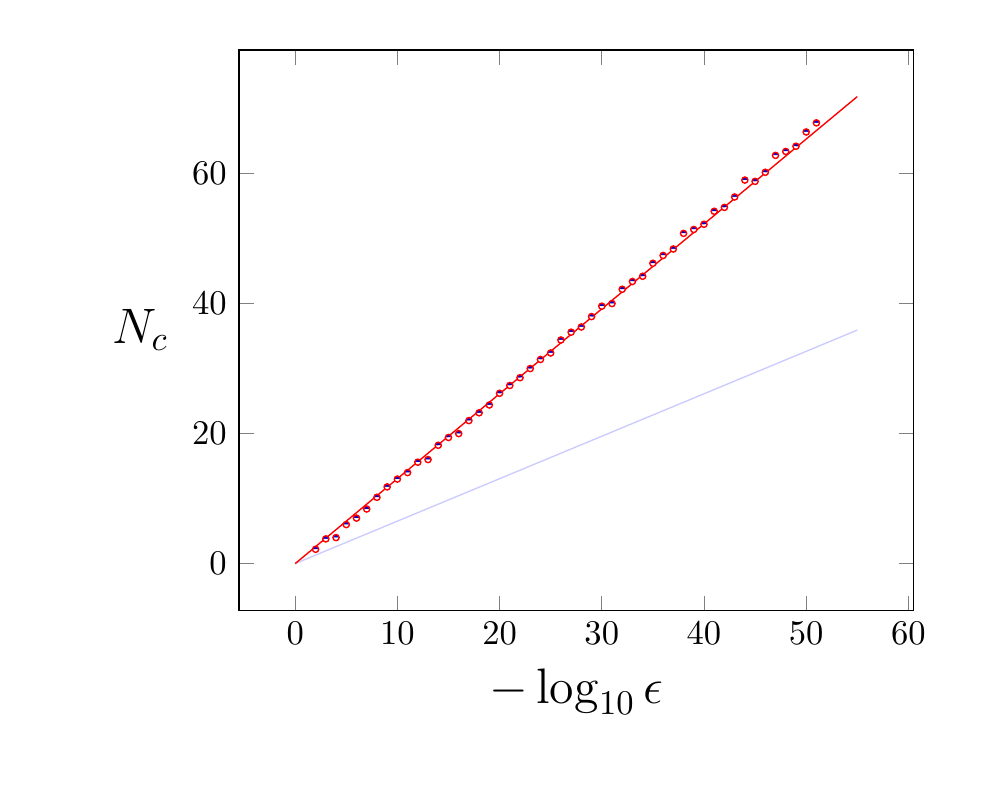}}
    \caption{This plot shows [red dots] the critical truncation order $N_c$ at which Pad\'e breaks down in the presence of noise, as a function of the logarithm of the noise strength, for the Borel transform (\ref{eq:borel1}) of the tritronqu\'ee solution of the Painlev\'e I equation \cite{Costin:2019xql}. The red dots show the average of multiple  realizations of the random noise, and the red line shows the general two-cut scaling relation in (\ref{eq:ncM}) with $M=2$. The faint blue line shows the slope for the one-cut scaling relation, for contrast.}
    \label{fig:painleve1-noise}
    \end{figure}

The Borel transform can be defined as \cite{Costin:2019xql}
\begin{eqnarray}
 B(\omega)= \sum_{n=1}^{\infty} \frac{a_n}{(2n-1)!}\, \omega^{2n-1}
\label{eq:borel1}
\end{eqnarray}
The expansion coefficients $a_n$ are rational numbers, generated from the recursion relation
\begin{eqnarray}
a_n&=&-4(n-1)^2a_{n-1}-\frac{1}{2}\sum_{m=2}^{n-2} a_m\, a_{n-m}\qquad, \quad n\geq 3 
\nonumber\\
&& a_1=\frac{4}{25}\qquad , \qquad a_2=-\frac{392}{625}
\label{eq:recursion}
\end{eqnarray}
In the absence of noise,  finite-order truncations of this Borel transform encode non-trivial information about the global analytic properties of the tritronqu\'ee solution \cite{Costin:2019xql}. When we introduce noise, as in (\ref{eq:add-noise}), then we observe that the Pad\'e approximation to $B_\epsilon (\omega)$ breaks down at a truncation order $N_c$ that depends on the strength of the noise. Figure \ref{fig:painleve1-noise} plots this critical truncation order $N_c$ as a function of the logarithm of the noise. Once again we see that consideration of the effect of the {\it dominant} pair of branch point Borel singularities  is sufficiently accurate to describe the effect of noise on this expansion.
    
\section{Conclusions}

We have analyzed the effect of noise on Pad\'e approximants, for functions with general branch point singularities of the form that arise in a broad class of physical applications. With noisy input coefficients, the Pad\'e approximant breaks down at a certain Pad\'e order, $N_c$, which is proportional to the log of the noise strength. Furthermore, the proportionality constant  can be expressed in terms of the conformal map that Pad\'e generates in its large order limit. We have presented two natural ways to characterize the breakdown of Pad\'e: one based on a sudden change in the distribution of the Pad\'e poles, and another based on a change in the relative precision of the Pad\'e approximant. Our main results are Theorem \ref{T:4.1} and Corollary \ref{C:4.7}, which characterize the breakdown condition both globally and locally.
For a given level of noise there is an order beyond which Pad\'e will begin to introduce spurious poles that do not represent the true singularity structure of the function being approximated. Correspondingly, the extrapolation accuracy of Pad\'e will degrade beyond this threshold breakdown order. Theorem \ref{T:4.1} and Corollary \ref{C:4.7} furthermore identify the locations at which spurious poles form. 
The proportionality constant relating the breakdown order to the logarithm of the noise strength can be expressed in terms of the conformal map that Pad\'e generates at large order. Therefore this slope can be estimated just based on the locations of the singularities, not requiring full information about the function itself. Furthermore, we have shown that in some non-trivial problems the slope can be accurately estimated based solely on the effect of the {\it dominant} singularities, not even requiring knowledge of the exact conformal map.  We anticipate that this result can have implications in a wide range of physical applications. An important open question is to determine {\it optimal} strategies for extrapolation in the presence of noisy coefficients, generalizing the results of \cite{Costin:2020pcj} for the noise-free case.

 \vspace{.3cm}
\noindent {\bf Acknowledgements} \\
This work is supported in part by the U.S. Department of Energy, Office of High Energy Physics, Award  DE-SC0010339 (GD, MM), and by the U.S. National Science Foundation, Division of Mathematical Sciences, Award NSF DMS - 2206241 (OC).


\begin{thebibliography}{99}

  
\bibitem{baker}
G. A. Baker and P. Graves-Morris,
{\it Pad\'e Approximants}, (Cambridge University Press, 2009).

\bibitem{bender}
C. M. Bender and S. A. Orzsag,
{\it Advanced Mathematical Mehtods for Scientists and Engineers}, (Springer, 1999).

 \bibitem{Stahl} H. Stahl, 
  ``The Convergence of Pad\'e Approximants to Functions with Branch Points'', 
  J. Approx. Theory {\bf 91}, 139-204 (1997).

\bibitem{Saff} E. B. Saff, 
``Logarithmic Potential Theory with Applications to Approximation Theory'', 
Surveys in Approximation Theory {\bf 5},  165-200 (2010), 
\hhref{1010.3760}.

\bibitem{gonchar}
A. A. Gonchar, E. A. Rakhmanov, and S. P. Suetin,
``Pad\'e-Chebyshev approximants of multivalued
analytic functions, variation of equilibrium energy,
and the S-property of stationary compact sets'',
Russian Math. Surveys {\bf 66}:6, 1015-1048 (2011).

 \bibitem{aptekarev}
 A. I. Aptekarev, V. I. Buslaev, A. Mart\'inez-Finkelshtein, and S. P. Suetin   
``Pad\'e approximants, continued fractions, and orthogonal polynomials'',
Russian Math. Surveys {\bf 66}:6 1049-1131 (2011).

 
\bibitem{Fink} A. Mart\'inez-Finkelshtein, E. A. Rakhmanov, S. P. Suetin, 
``Heine, Hilbert, Pade, Riemann, and Stieltjes: John Nuttall's work 25 years later'', 
Contemporary Mathematics 578, 165-193 (2012).

\bibitem{Yattselev} A. Aptekarev and M. L. Yattselev, 
  ``Pad\'e approximants for functions
    with branch points - strong asymptotics of Nuttall-Stahl polynomials'', 
    Acta Math. {\bf 215},  217-280 (2015).
    

\bibitem{Costin:2021bay}
O.~Costin and G.~V.~Dunne,
``Conformal and uniformizing maps in Borel analysis,''
Eur. Phys. J. ST \textbf{230}, no.12-13, 2679-2690 (2021),
\hhref{2108.01145} [hep-th].

    
    \bibitem{froissart}
M. Froissart, 
``Approximation de Pad\'e. Application \`a la physique des particules el\'ementaires'',
Les rencontres physiciens-math\'ematiciens de Strasbourg RCP25 {\bf 9} (1969).

\bibitem{bessis1}
D. Bessis,
``Pad\'e approximations in noise filtering'',
Journal of Computational and Applied Mathematics {\bf 66} (1996) 85-88.

\bibitem{bessis2}
D. Bessis and L. Perotti,
``Universal analytic properties of noise: introducing the J-matrix formalism'', 
J. Phys. A: Math. Theor. {\bf 42} (2009) 365202.

\bibitem{gilewicz1}
J. Gilewicz and M. Pindor, 
``Pad\'e approximants and noise: a case of geometric
series'', Journal of Computational and Applied Mathematics {\bf 87}, 199-214 (1997);
``Pad\'e approximants and noise: rational functions'', 
Journal of Computational and Applied Mathematics {\bf 105}, 285-297 (1999).


\bibitem{gilewicz2}
J. Gilewicz and Y. Kryakin, 
``Froissart doublets in Pad\'e approximation in the case
of polynomial noise'', 
Journal of Computational and Applied Mathematics {\bf 153}, 235-242 (2003).
Proceedings of the 6th International Symposium on Orthogonal Polynomials, Special Functions and their Applications, Rome, Italy, 18-22 June 2001.

\bibitem{ZinnJustin:2002ru} 
  J.~Zinn-Justin,
 {\it Quantum Field Theory and Critical Phenomena}, 
  Int.\ Ser.\ Monogr.\ Phys.\  {\bf 113}, 1 (2002).
  
\bibitem{Costin:2020hwg}
O.~Costin and G.~V.~Dunne,
``Physical Resurgent Extrapolation,''
Phys. Lett. B \textbf{808}, 135627 (2020),
\hhref{2003.07451} [hep-th].

\bibitem{yamada}
H. S. Yamada and K. S. Ikeda,
"A Numerical Test of Pad\'e Approximation for Some Functions with Singularity'',
\hhref{1308.4453}.

        \bibitem{leguillou}
    J. C. Le Guillou and J. Zinn-Justin (Eds.),
    {\it Large order behavior of perturbation theory}, (North-Holland, 1990).
 
 \bibitem{kuzmina}
 G.V. Kuz'mina, 
 ``Estimates for the transfinite diameter of a family of continua and covering theorems for univalent functions'',
 Proc. Steklov Inst. Math. {\bf 94}, 53-74 (1969).
 
  \bibitem{grassmann}
E. G. Grassmann and J. Rokne, 
``An explicit calculation of some sets of minimal capacity'',
SIAM J. Math. Anal. {\bf 6}, 242-249 (1975).

\bibitem{ransford}
T. Ransford, 
``Computation of logarithmic capacity'', 
Computational Methods and
Function Theory {\bf 10},  555-578 (2011).


    
    \bibitem{churchill}
    J. W. Brown and R. V. Churchill,
    {\it Complex Variables and Applications}, 8th Edition (McGraw-Hill, 2009).
    
    
      \bibitem{Borel1}  \'Emile Borel, ``Sur les s\'eries de Taylor'', C.R. de l'Acad., 1896.
      
      \bibitem{kahane}
J-P. Kahane, 
{\it Some Random Series of Functions}, Second Edition. Cambridge University Press, 1985, pp.38-40.

 \bibitem{Ryll} Czes\l aw Ryll-Nardzewski, 
 ``D. Blackwell's conjecture on power series with random coefficients'' (Studia Math., 1953).
 
\bibitem{Costin:2020pcj}
O.~Costin and G.~V.~Dunne,
``Uniformization and Constructive Analytic Continuation of Taylor Series,''
Commun. Math. Phys. \textbf{392}, 863-906 (2022),
\hhref{2009.01962} [math.CV].

\bibitem{Damanik-Simon} D. Damanik and B. Simon, 
``Jost functions and Jost solutions for Jacobi matrices,
I. A necessary and sufficient condition for Szeg\"o 
asymptotics'', Invent. Math. {\bf 165}, 1-50 (2006).
    
\bibitem{Kober} H. Kober, {\it Dictionary of Conformal Representations}, Dover (1957).  

 
\bibitem{Lipatov}
L.~N.~Lipatov,
``Divergence of the Perturbation Theory Series and the Quasiclassical Theory,''
Sov. Phys. JETP \textbf{45}, 216-223 (1977).

\bibitem{Fisher}
M.~E.~Fisher,
``Yang-Lee Edge Singularity and phi**3 Field Theory,''
Phys. Rev. Lett. \textbf{40}, 1610-1613 (1978).

\bibitem{Mckane}
A.~J.~McKane,
``Vacuum Instability in Scalar Field Theories,''
Nucl. Phys. B \textbf{152}, 166-188 (1979).


\bibitem{Bonfim}
O.~F.~de Alcantara Bonfim, J.~E.~Kirkham and A.~J.~McKane,
``Critical Exponents for the Percolation Problem and the Yang-lee Edge Singularity,''
J. Phys. A \textbf{14}, 2391 (1981).

\bibitem{Borinsky}
M.~Borinsky, J.~A.~Gracey, M.~V.~Kompaniets and O.~Schnetz,
``Five-loop renormalization of $\ensuremath{\phi}^3$ theory with applications to the Lee-Yang edge singularity and percolation theory,''
Phys. Rev. D \textbf{103}, no.11, 116024 (2021),
\hhref{2103.16224} [hep-th].
  
\bibitem{Connes:1999yr}
A.~Connes and D.~Kreimer,
``Renormalization in quantum field theory and the Riemann-Hilbert problem. 1. The Hopf algebra structure of graphs and the main theorem,''
Commun. Math. Phys. \textbf{210}, 249-273 (2000),
\hhref{hep-th/9912092} [hep-th];
``Renormalization in quantum field theory and the Riemann-Hilbert problem. 2. The beta function, diffeomorphisms and the renormalization group,''
Commun. Math. Phys. \textbf{216}, 215-241 (2001),
\hhref{hep-th/0003188} [hep-th].

\bibitem{Broadhurst:1999ys}
D.~J.~Broadhurst and D.~Kreimer,
``Combinatoric explosion of renormalization tamed by Hopf algebra: Thirty loop Pade-Borel resummation,''
Phys. Lett. B \textbf{475}, 63-70 (2000),
\hhref{hep-th/9912093} [hep-th];
``Exact solutions of Dyson-Schwinger equations for iterated one loop integrals and propagator coupling duality,''
Nucl. Phys. B \textbf{600}, 403-422 (2001),
\hhref{hep-th/0012146} [hep-th].
  
\bibitem{Borinsky:2021hnd}
M.~Borinsky, G.~V.~Dunne and M.~Meynig,
``Semiclassical Trans-Series from the Perturbative Hopf-Algebraic Dyson-Schwinger Equations: $\phi^3$ QFT in 6 Dimensions,''
SIGMA \textbf{17}, 087 (2021),
\hhref{2104.00593} [hep-th].

\bibitem{Borinsky:2022knn}
M.~Borinsky and D.~J.~Broadhurst,
``Resonant resurgent asymptotics from quantum field theory,''
Nucl. Phys. B \textbf{981}, 115861 (2022), 
\hhref{2202.01513} [hep-th].

\bibitem{clarkson}
P. A. Clarkson, 
``Painlev\'e Equations - Nonlinear Special Functions'', in {\it Orthogonal Polynomials and Special Functions}, F. Marcell\'an and W. Van Assche (eds), Lecture Notes in Mathematics, vol 1883 (Springer, Berlin).

  \bibitem{costin-book}
O. Costin,
{\it Asymptotics and Borel summability},
(Chapman and Hall/CRC, 2008).

\bibitem{DiFrancesco:1993cyw}
P.~Di Francesco, P.~H.~Ginsparg and J.~Zinn-Justin,
``2-D Gravity and random matrices,''
Phys. Rept. \textbf{254}, 1-133 (1995),
\hhref{hep-th/9306153} [hep-th].

\bibitem{Costin:2019xql} 
  O.~Costin and G.~V.~Dunne,
 ``Resurgent extrapolation: rebuilding a function from asymptotic data. Painlev\'e I,''
  J.\ Phys.\ A {\bf 52}, no. 44, 445205 (2019),
 \hhref{1904.11593} [hep-th].

\bibitem{dubrovin}
B. Dubrovin, T. Grava, and C. Klein, 
``On universality of critical behavior in the focusing nonlinear Schr\"odinger
equation, elliptic umbilic catastrophe and the tritronqu\'ee solution to the Painlev\'e-I equation'', 
J. Nonlinear Sci. {\bf 19}, 57-94 (2009).

\bibitem{costin-dubrovin}
 O. Costin, M. Huang and S. Tanveer,
 ``Proof of the Dubrovin conjecture and analysis of the tritronqu\'ee solutions of PI'',
  Duke Math. J. {\bf 163} (4),  665-704 (2014).
  
 
\end{thebibliography}
\end{document}